\documentclass[journal]{IEEEtran}
 
\usepackage{color}
\usepackage{amsfonts}
\usepackage{amsmath}
\usepackage{graphicx}

\usepackage{cite}

\usepackage{epsfig}
\usepackage{latexsym}  
\usepackage{subfigure}
\usepackage{epstopdf}
\usepackage{color}
\usepackage{verbatim}
\usepackage{units}
\usepackage{amsthm}
\usepackage[longend,ruled]{algorithm2e}
\usepackage{booktabs}

\newtheorem{lemma}{Lemma}

\usepackage{amsmath}

\usepackage{amsthm}
\newtheorem*{remark}{Remark}

\usepackage{color,soul}

\usepackage{amsmath,amssymb,enumitem,mathtools}
\usepackage{xcolor}

\usepackage[most]{tcolorbox}
\newtcolorbox{highlighted}{colback=yellow,coltext=red,breakable}
\usepackage{lipsum}

\usepackage{bbm}
\usepackage{dsfont}

\begin{document}
\title{Inter-cluster Cooperation for Wireless D2D Caching Networks}
 
\author{Ramy Amer,~\IEEEmembership{Student~Member,~IEEE,} M.~Majid~Butt,~\IEEEmembership{Senior~Member,~IEEE,}  Mehdi~Bennis,~\IEEEmembership{Senior~Member,~IEEE,}~and~Nicola~Marchetti,~\IEEEmembership{Senior~Member,~IEEE}
 
\thanks{The material in this paper is presented in part at Globecom 2017 \cite{Delay-Analysis}.}
\thanks{This publication has emanated from research conducted with the financial support of Science Foundation Ireland (SFI) and is co-funded under the European Regional Development Fund under Grant Number 13/RC/2077.}
}

\maketitle			
\begin{abstract} 
Proactive wireless caching and device to device (D2D) communication have emerged as promising techniques for enhancing users' quality of service and network performance. In this paper, we propose a new architecture for D2D caching with inter-cluster cooperation. We study a cellular network in which users cache popular files and share them with other users either in their proximity via D2D communication or with remote users using cellular transmission. We characterize the network average delay per request from a queuing perspective. Specifically, we formulate the delay minimization problem and show that it is NP-hard. Furthermore, we prove that the delay minimization problem is equivalent to the minimization of a non-increasing monotone supermodular function subject to a uniform partition matroid constraint. A computationally efficient greedy algorithm is proposed which is proven to be locally optimal within a factor $(1 - e^{-1})\approx 0.63$ of the optimum. We analyze the average per request throughput for different caching schemes and conduct the scaling analysis for the average sum throughput. We show how throughput scaling depends on video content popularity when the number of files grows asymptotically large. Simulation results show a delay reduction of 45\% to 80\% compared to a D2D caching system without inter-cluster cooperation.

\end{abstract}

\IEEEpeerreviewmaketitle

\section{Introduction}
The rapid proliferation of mobile devices has led to unprecedented growth in wireless traffic demands. A typical approach to deal with such demand is by densifying the network. For example, macrocells and femtocells are deployed to enhance the capacity and attain a good quality of service (QoS) by bringing the network closer to the user. Recently, it has been shown that only a small portion of multimedia content is highly demanded by most of the users. This small portion forms the majority of requests that come from different users at different times, which is referred to as $\textit {asynchronous content reuse}$ \cite{mono}. 

Caching the most popular content at various locations of the network edge has been proposed to avoid serving all requests from the core network through highly congested backhaul links
 \cite{wr,Living_on_the_edge,Cache-enabled-small-cell}. From the caching perspective, there are three main types of networks, namely, caching on femtocells in small cell networks, caching on remote radio heads (RRHs) in cloud radio access networks (RANs), and caching on mobile devices \cite{femtocell_mehdi, ran, D2D}. One approach to overcome the limitations of the finite capacity backhaul links in the Cloud-RANS, where low energy base stations (BSs) are deployed over a small geographical area and are connected to the cloud, is to introduce local storage caches at the BSs, in which the popular files are stored locally in order to reduce the load of the backhaul links \cite{ran}. 
For the small cell networks, caching the most popular content at the network edge (the small BSs) is a promising solution to reduce the traffic and the energy consumption over the finite capacity backhaul links \cite{femtocell_mehdi}.

In this article, we focus on device caching solely. The architecture of device caching exploits the large storage available in modern smartphones to cache multimedia files that might frequently be requested by the users. The users' devices exchange multimedia content stored on their local storage with nearby devices \cite{D2D}. Since the distance between the requesting user and the caching user (a user who stores the file) will be small in most cases, device to device (D2D) communication is commonly used for content transmission \cite{D2D}.
In this context, Golrezaei \textit{et al.} \cite{D2D1} proposed a novel architecture to improve the throughput of video transmission in cellular networks based on the caching of popular video files in base station controlled D2D communication. The analysis of this network is based on the subdivision of a macrocell into small virtual clusters, such that one D2D link can be active within each cluster. Random caching is considered where each user caches files at random and independently, according to a caching distribution.

Different cooperation strategies in D2D networks are proposed in the literature. As an example, in \cite{6952682}, the authors proposed a cooperative D2D communications framework in order to combat the problem of congestion in crowded communication environments. The authors allowed a D2D transmitter to act as an in-band relay for a cellular link and at the same time transmit its data by employing superposition coding in the downlink. It is shown that cooperation between the cellular link and D2D transmitter helps increase the number of connections per unit area with the same spectrum usage. In the area of D2D caching, the authors in \cite{D2D3} proposed an opportunistic cooperation strategy for D2D transmission by exploiting the caching capability at the users to control the interference among D2D links.
The authors considered an overlay inband D2D communication, divided the D2D users into clusters, and assigned different frequency bands to cooperative and non-cooperative D2D links. The cluster size and bandwidth allocation are further optimized to maximize the network throughput.

The analysis of wireless caching networks from the resource allocation perspective is widely discussed in the literature. For instance, in \cite {Wireless-video-content}, the authors showed how distributed caching and collaboration between users and femtocells (helpers) can significantly improve throughput without suffering from the backhaul bottleneck problem common to femtocells. The authors also investigated the role of collaboration among users - a process that can be interpreted as the mobile devices playing the role of helpers also. This approach allowed an improvement in the video throughput without the deployment of any additional infrastructure. 
Due to the dependence between content cache placement and resource allocation in wireless networks, the joint problem of caching and resource allocation is studied in many works. As an example, Zhang \textit{et al.} in \cite {Efficient_Scheduling} proposed a single-hop D2D-assisted wireless caching network, where popular files are randomly and independently cached in the memory of end users. The joint D2D link scheduling and power allocation problem is formulated to maximize the system throughput. 
Following a similar approach, Chen \textit{et al.} in \cite{D2D2} studied the joint optimization of cache content placement and scheduling policies to maximize the so-called offloading probability. 
The successful offloading probability is defined as the probability that a user can obtain the desired file in the local cache or via a D2D link with data rate larger than a given threshold. The authors obtained the optimal scheduling factor for a random scheduling policy that controls interference in a distributed manner and proposed a low complexity solution to
compute caching distribution.

Motivated by the remarks from the above discussion, i.e., backhaul links being highly congested, the geometric distribution of the users as groups in clusters, and the small memory sizes of a group of users colocated in the same cluster, we propose a novel D2D caching architecture with inter-cluster cooperation. We propose a system in which a user in a given cluster can search its requested files either in the local cluster or any of the remote clusters. We show that allowing inter-cluster collaboration via cellular communication achieves both user and system performance gains. From the user perspective, the average delay per request is reduced when downloading files from a remote cluster instead of serving files from the core network. From the system perspective, the heavy burden on backhaul links is alleviated by decreasing the number of requests that are served directly from the core network. 
 From a resource allocation perspective, similar to the work performed in \cite{Wireless-video-content, D2D3}, we analyze the network average delay and throughput per user request for the proposed inter-cluster cooperative caching system under different caching schemes and show how the network performance is significantly improved. 
To the best of our knowledge, none of the works in the literature dealt with the performance analysis of D2D caching networks with inter-cluster cooperation.

The main contributions of this article are summarized as follows:
\begin{itemize}
\item We study a D2D caching system with inter-cluster cooperation from a queueing theory perspective. We formulate the network average delay minimization problem in terms of cache placement. The delay minimization problem is then shown to be non-convex, and it can be reduced to a well-known 0 - 1 knapsack problem which is NP-hard.
\item A closed-form expression of the network average delay is derived under the policy of caching popular files. Moreover, a locally optimal greedy caching algorithm is proposed whose delay is within a factor $(1 - e^{-1})$ of the global optimum. Results show that the delay can be significantly reduced by allowing D2D caching with inter-cluster cooperation. 
\item We derive a closed form expression for the average throughput per request for the proposed inter-cluster cooperating scheme. Moreover, we conduct the asymptotic analysis for the average sum throughput when the content library size grows to infinity.	
The result of the scaling analysis shows that the upper bound for the network average sum throughput decreases when the library size increases asymptotically, and the rate of this decrease is controlled by the popularity of files.

\end{itemize} 

The rest of the paper is organized as follows. The system model is presented in Section II. In Section III, we formulate the problem and perform the delay analysis of the system. In Section IV, the content caching schemes are studied. Section V provides the throughput analysis. Finally, we discuss the simulation and analytical results in Section VI and conclude the paper in Section VII.

\section{System Model}
\label{sysmodel}
\begin{figure}
\centering
\includegraphics[width=0.5\textwidth]{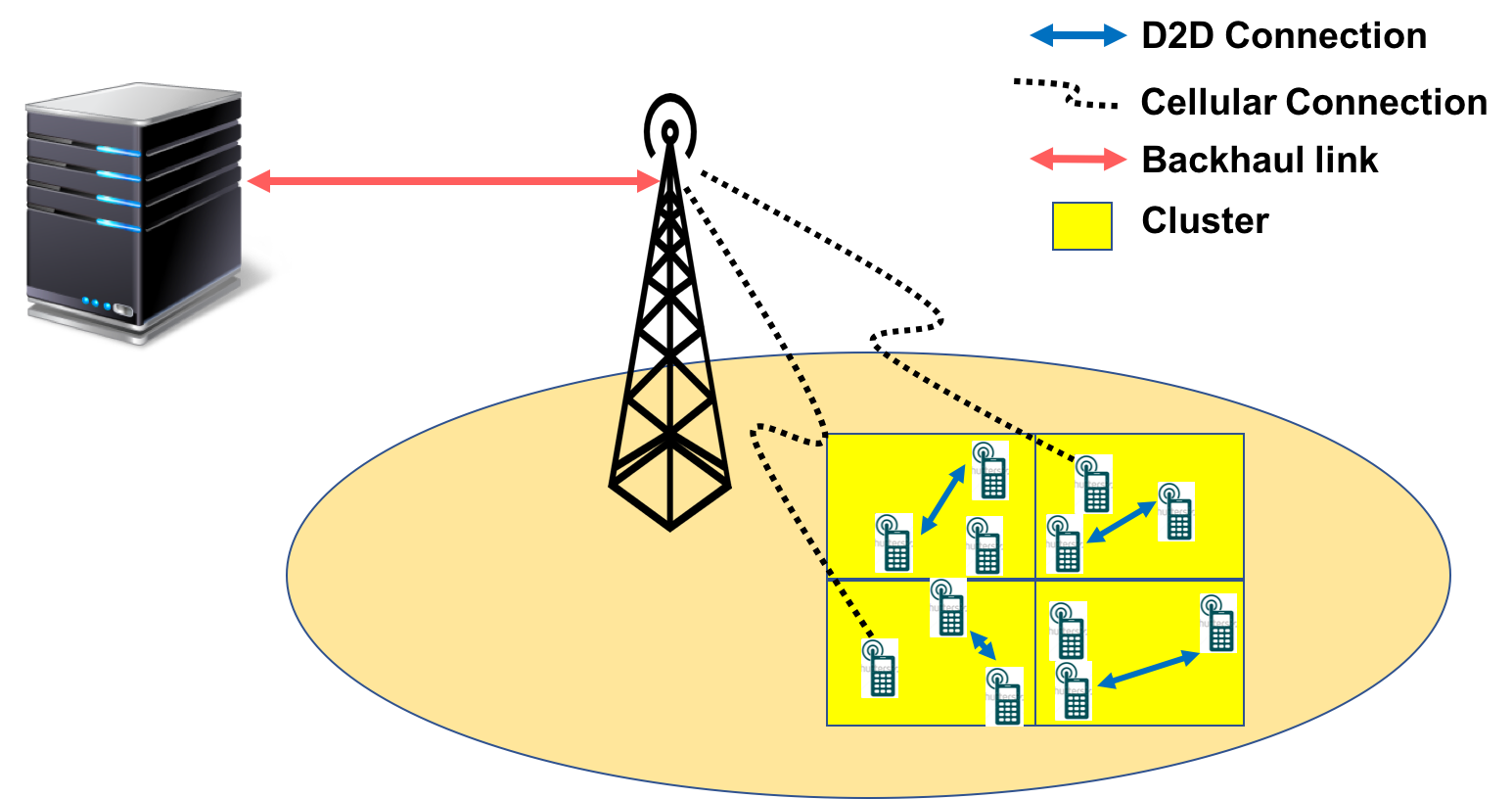}				
\caption {Schematic diagram of the proposed system model. A cellular cell is divided into square clusters, where users in all clusters can download their requested files using D2D, cellular, or backhaul communication.}
\label{Network Model}
\end{figure}

\subsection{Network Model}
In this subsection, we describe our proposed D2D caching network with inter-cluster cooperation. Fig.~\ref{Network Model} illustrates the system layout. A cellular network consists of a small base station (SBS) and a set of users $\mathcal{U}= \{1,\dots, n\}$ placed uniformly in the cell. The cell is divided into a set of equally sized clusters $\mathcal{K}= \{1,\dots, K\}$. For mathematical convenience, we assume that the number of users per cluster is $y = n/K$ users, as in \cite{D2D3} and the reference therein. Users in the same cluster can communicate directly using low power high rate D2D communication in a dedicated frequency band for D2D transmission.

Each user $u \in \mathcal{U}$ requests a file $f$ from a file library $\mathcal{F}= \{1,\dots, m\}$ independently and identically, according to a given request probability mass function. It is assumed that each user can cache up to $M$ files, and for the caching problem to be non-trivial, it is assumed that $M < m$. From the cluster perspective, we assume to have a cluster's virtual cache center (VCC) formed by the union of devices' storage in the same cluster, which caches up to $N$ files, i.e., $N = (n/K)M$.

We assume that the D2D communication does not interfere with communication between the BS and users. We also assume that all D2D links share the same time-frequency transmission resource within one cell. Multiple transmissions on those resources are possible since the distance between requesting users and users with the stored file will typically be small. Furthermore, there should be no interference by other transmissions on an active D2D link. To achieve this, the cell is divided into smaller areas, which we denoted as clusters. To avoid intra-cluster interference, only one such communication per cluster is allowed.\footnote{We adopt a simplified PHY-layer model in this work.} Users in the same cluster are assumed to be served in a round-robin manner.

 We define three modes of operation according to how a request for content $ f \in \mathcal{F}$ is served:
\begin{enumerate}
\item {\textbf{Local cluster mode ($M_{lc}$ mode):}} Requests are served from the local cluster.  Files are downloaded from nearby users via a single-hop D2D communication. In this mode, we neglect self-caching, i.e., the event when a user finds the requested file in its internal cache with zero delay. Within each cluster, the BS can help devices find their requested content by broadcasting signals containing the content replication ratio.
\item \textbf{Remote cluster mode ($M_{rc}$ mode):} Requests are served from any of the remote clusters via inter-cluster cooperation. The BS fetches the requested content from a remote cluster, then delivers it to the requesting user by acting as a relay in a two-hop cellular transmission. The BS assists in content dissemination in the {\em ``remote cluster mode"} by relaying the content between different clusters.
\item \textbf{Backhaul mode ($M_{bh}$ mode):} Requests are served directly from the backhaul. The BS obtains the requested file from the core network via the backhaul link and then transmits it to the requesting user.
\end{enumerate}

In each cluster, we assume that the stream of user requests are served sequentially based on first in first out (FIFO) criterion. The BS receives all requests and works as a coordinator to establish the file transfer between the requesting user (a user who requests the file) and the serving node (another user who caches the file or a caching server in the core network).
The BS keeps track of which devices can communicate with
each other and which files are cached on each device. Such
BS-controlled D2D communication is more efficient and more
acceptable to spectrum owners if the communication occurs in
a licensed band as compared to traditional uncoordinated peer-to-peer communications \cite{caire2015femtocaching}. To serve a request for file $f$ in cluster $k \in \mathcal{K}$, first, the BS searches the VCC of cluster $k$. If the file is cached, it will be delivered from the local VCC ($M_{lc}$ mode). We assume that the BS has all the information about cached content in all clusters, such that all file requests are sent to the BS, then the BS replies with the address of the caching user from whom the file will be retrieved. 

If a file is not cached locally in cluster $k$ but cached in any of the remote clusters, it will be fetched from a randomly chosen cooperative cluster ($M_{rc}$ mode), instead of downloading it from the backhaul. Unlike multi-hop D2D cooperative caching discussed in \cite{7820112}, in our work cooperating clusters are assumed to exchange cached files using a  two-hop cellular communication link through the BS, such that the D2D band is dedicated only to the intra-cluster communication. Hence, all the inter-cluster communication is performed in a centralized manner through the BS. Finally, if the requested file has not been cached in any cluster $j \in \mathcal{K}$ in the cell, it can be downloaded from the core network via the backhaul link ($M_{bh}$ mode).
 The selection of the three modes of operation is conducted in a prioritized order from the local cluster, from the remote cluster, or finally from the core network through the backhaul link as a last resort. 

Serving files sequentially according to the above three modes is based on the assumption that the BS has a capacity limited wired backhauling, such that the average delay per request is decreased when allowing inter-cluster cooperation.
Otherwise, if the backhaul is not a bottleneck, e.g., optical fiber or millimeter wave backhaul links are available, requests for files not cached in the local cluster are served directly from the core network through the high capacity backhaul link. The analysis in this paper relies on a well-known grid-based clustering model \cite{D2D1}, i.e., no specific underlying physical model or parameters are assumed. Therefore, the obtained design/results, e.g., design of caching scheme and the performance of the greedy algorithm, can be applied to similar scenarios with three prioritized paths (modes) for file downloading. For example, on-board users, such as on a plane or a ship, can obtain requested files from neighboring users via Bluetooth (local cluster mode), from a remote user through an access point \cite{altman2013coding} acting as a relay (remote cluster mode), or finally from the backhaul, which is the least preferred option. As another example, in the case of connecting users through unmanned aerial vehicles (UAVs) \cite{sundaresan2018skylite}, serving files can be prioritized as follows. A file is received from a neighboring user via D2D communication (local cluster mode), from a remote user through the UAV acting as a relay (remote cluster mode), or from the backhaul to the core network through the UAV as a last resort.

\subsection{Content Placement and Traffic Characteristics}
\label{sec}
We use a binary matrix $\textbf{C}= [c_{k,f}]_{K\times m}$ with $c_{k,f} \in \{0, 1\}$ to denote the cache placement in all clusters, where $c_k,_f = 1$ indicates that content $f$ is cached in cluster  $k$. Fig.~\ref{vcc}  shows the assumed users' traffic model in a cluster $k$, modeled as a \textit{multiclass processor sharing queue} (MPSQ) with arrival rate $\lambda_k$, and three serving processors representing the three transmission modes. According to the MPSQ definition \cite{MPSQ}, each transmission mode is represented by an M/M/1 queue with Poisson arrival rate and exponential service rate. A graphical interpretation of the content cache placement is shown in Fig.~\ref{bi}. The content caching policy is defined by a bipartite graph $\mathcal{Y} = (\mathcal{K},\mathcal{F}, \mathcal{E})$, where edges $(k,f)\in \mathcal{E}$ denote that content $f$ is cached in the VCC of cluster $k$. 

\begin{figure}
\begin{center}
\includegraphics[width=0.3\textwidth]{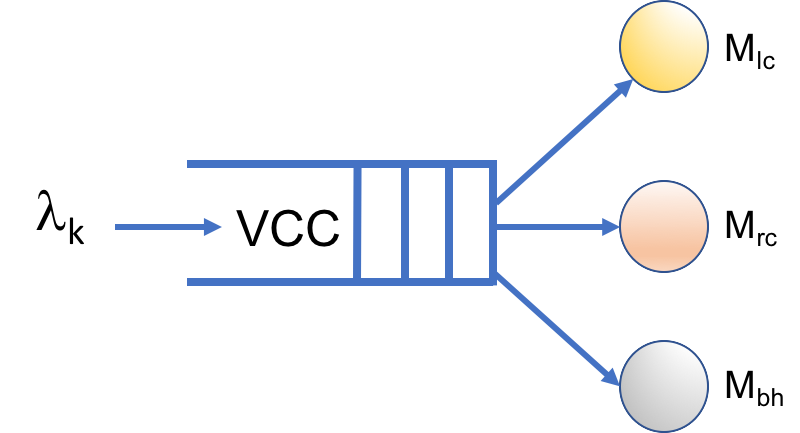}			
\caption {The users' traffic model in a cluster $k$, with a cache center VCC, is modeled as a \textit{multiclass processor sharing queue} (MPSQ).}
\label{vcc}
\end{center}
\end{figure}
		
\begin{figure}
\begin{center}
\includegraphics[width=0.4\textwidth]{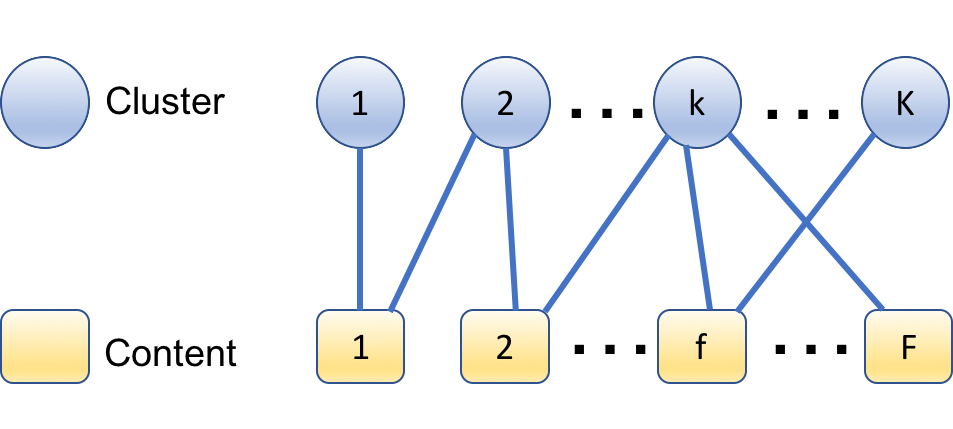}
\caption {An example of the content cache placement modeled as a bipartite graph indicating how files are cached in clusters.}
\label{bi}
\end{center}
\end{figure}

If a user in cluster $k$ requests a locally cached file $f$ (i.e., $c_k,_f = 1$), it will be served by the local cluster mode with an average rate $R_{D}$. However, if the requested file is not cached locally and cached in any of the remote clusters, i.e., when $c_k,_f = 0$  and $\sum_{j\in \mathcal{K}\setminus \{k\}}c_j,_f \geq 1$, it will be served by the remote cluster mode.

We denote the rate for the remote cluster mode by $R_{WL}$, accounting
for the average sum transmission rate between the cooperating
clusters through the BS. Accordingly, $R_{WL}$ is shared between clusters simultaneously served by the remote cluster mode. Finally, requests for files that are not cached in the entire cell, i.e., when $\sum_{j=1}^{K}  c_k,_f = 0$, are served via the backhaul mode with an average sum rate $R_{BH}$. We assume that $R_{BH} << R_{WL}$, such that the part of the cellular rate allocated to the users served by the backhaul mode is neglected for the delay analysis. $R_{BH}$ is assumed to be the effective rate from the core network to the user using BS.

Due to traffic congestion in the core network and the transmission delay between cooperating clusters, we assume that the aggregate transmission rates for the above three modes are ordered such that $R_D > R_{WL} >  R_{BH}$. We also assume that the content size $S_f$ is exponentially distributed with mean $\overline{S}$  bits. Hence, the corresponding request service times of the three transmission modes also follow an exponential distribution with means $\tau_{lc} = \frac{\overline{S}}{R_{D}}$ sec, $\tau_{rc} = \frac{\overline{S}}{R_{WL}}$ sec, and $\tau_{bh} = \frac{\overline{S}}{R_{BH}}$ sec, respectively.
\section{Problem Formulation}
\label{prob}
In this section, we characterize the network average delay on a per request basis from the global network perspective.
Specifically, we study the request arrival rate and the traffic dynamics from a queuing theory perspective and get a closed form expression
for the network average delay.
 \subsection{File Popularity Distribution}
We assume that the popularity distribution of files in all clusters follows a Zipf's distribution with skewness order $\beta$ \cite{zipf}. However, it is assumed that the content may vary across clusters. This is inspired by the fact that, for instance, users in a library may be interested in an entirely different set of files from the users in a sports center.
Our assumption for the popularity distribution is extended from \cite{basic_principle}, where the authors explained that the scaling of popular files is sublinear with the number of users.\footnote
{The number of popular files increases with the number of users with a rate slower than the linear polynomial rate, e.g., the logarithmic rate.}

  To illustrate, if user 1 and user 2 are interested in a set of files with size $m_0$, then the first $m_0/2$ files of user 2 are common with user 1 and $m_0/2$ are the new ones. User 3, in turn, shares $2m_0/3$ files with users 1 and 2, and has $m_0/3$ new files, etc. The union of all demanded (popular) files by $n$ users is $m = m_0 (1 + \frac{1}{2} + \frac{1}{3} + \dots)= m_0 \sum_{i=1}^{n}\frac{1}{i} \approx m_0$ log $n$. Hence, the library size increases sublinearly with the number of users. In this work, we assume that the scaling of the library size is sublinear with the number of clusters. The cell is divided into clusters with a small number of users per cluster, such that users in the same cluster are assumed to request files according to the same file popularity distribution function (i.e., users in the same cluster are interested in the same set of popular files).

The probability that a file $f$ is requested in cluster $k$, with $m_0$ highly demanded files in each cluster, follows a Zipf distribution written as \cite{zipf},
\begin{equation}
P_k,_f = \frac{(  f - \frac{k-1}{k}m_0 a  +  (m - \frac{k-1}{k}m_0)b  )^{-\beta}}{\sum_{i=1}^{m}i^{-\beta}},
\label{popularity eqn}
\end{equation}
where $a = \mathds{1}(f>\frac{k-1}{k}m_0)$ and $b = \mathds{1}(f \leq\frac{k-1}{k}m_0)$, $\frac{k-1}{k}m_0$ is the order of the most popular file in the $k-$th cluster, and $\mathds{1}(.)$ is the indicator function.
When $k=1$, we get $P_1,_f = \frac{(  f  )^{-\beta}}{\sum_{i=1}^{m}i^{-\beta}}$ for the first cluster, which is the Zipf's distribution with the most popular file $f = 1$. For example, if $m_0 = 60$, then $P_2,_f = \frac{(  f  - 30a + (m - 30)b )^{-\beta}}{\sum_{i=1}^{m}i^{-\beta}}$ for the second cluster, which is the Zipf's distribution with the most popular file $f = \frac{m_0}{2} + 1 = 31$; also $f = \frac{2m_0}{3} + 1 = 41$ is the most popular file in the third cluster, and so on. 
\subsection{Arrival and Service Rates}
The arrival rates for the three communication modes $M_{lc}$, $M_{rc}$, and $M_{bh}$ in a cluster $k$ are denoted respectively by $\lambda_{k,lc}$, $\lambda_{k,rc}$, and $\lambda_{k,bh}$ while the corresponding service rates are represented by $\mu_{lc}$, $\mu_{rc}$, and $\mu_{bh}$. For the local cluster mode, we have
\begin{equation}
\lambda_k,_{lc} = \lambda_k \sum_{f=1}^{m}  P_k,_f c_k,_f,
\end{equation}
where $\sum_{f=1}^{m}  P_k,_f c_k,_f$ is the probability that the requested file is cached locally in cluster $k$. The corresponding service rate is $\mu_{lc} = \frac{1}{\tau_{lc}}$. For the remote cluster mode, the request arrival rate is defined as
\begin{equation}
\lambda_k,_{rc} = \lambda_k \sum_{f=1}^{m}  P_k,_f (1 - c_k,_f)\mathrm{min}\Big(\sum_{j \in \mathcal{K}\setminus \{k\}}c_{j,f},1\Big), 
\end{equation}
where min$(\sum_{j \in \mathcal{K}\setminus \{k\}}c_{j,f},1)$ equals one only if the content $f$ is cached in at least one of the remote clusters. Hence, $\sum_{f=1}^{m}  P_k,_f (1 - c_k,_f)$min$(\sum_{j \in \mathcal{K}\setminus \{k\}}c_{j,f},1)$ is the probability that the requested file $f$ is cached in any of the remote clusters given that it is not cached in the local cluster $k$. The corresponding service rate is $\mu_{rc} = \frac{1}{\tau_{rc} N_a}$, where $N_a$ represents the number of cooperating clusters simultaneously served by the remote cluster mode, i.e, the number of clusters which share the cellular rate.

Finally, for the backhaul mode, the request arrival rate is written as
\begin{equation}
\lambda_k,_{bh} = \lambda_k \sum_{f=1}^{m}  P_k,_f  \prod \limits_{k=1}^{K} (1 - c_k,_f),
\end{equation}
where $\sum_{f=1}^{m}  P_k,_f  \prod \limits_{k=1}^{K} (1 - c_k,_f)$ is the probability that the requested file $f$ is not cached entirely in the cell, so this content could be downloaded only from the core network. The corresponding service rate is $\mu_{bh} = \frac{1}{\tau_{bh} N_b}$, where $N_b$ is defined as the number of clusters simultaneously served via the backhaul mode.

The traffic intensity of a queue is defined as the ratio of mean service time to mean inter-arrival time. We introduce $\rho_k$ as a metric of the traffic intensity at cluster $k$ as 			

\begin{equation}
\rho_k = \frac{\lambda_k,_{lc}}{\mu_{lc}} +  \frac{\lambda_k,_{rc}}{\mu_{rc}}  +  \frac{\lambda_k,_{bh}}{\mu_{bh}}	
\label{rho_k}
\end{equation}
Similar to \cite{delayequation}, we consider $\rho_k < 1 $ as the stability condition, otherwise, the overall delay will be infinite. The traffic intensity at any cluster is simultaneously related to the request arrival rate and the transmission rates of the three serving modes.

\subsection{Network Average Delay}	
In \cite{delayequation}, it is proven that the mean queue size for an MPSQ with arrival rate $\lambda$ [sec$^{-1}$] and traffic intensity $\rho$, is 
\begin{equation}
\rho + \frac{\lambda\sum_{i}\frac{\lambda_i}{\mu_i^2}}{1 - \rho}, \nonumber
\end{equation}
where $\lambda_i$ and $\mu_i$ are respectively the arrival and service rates of a service group $i$. Given the fact that the average delay equals the mean queue size divided by the arrival rate, substituting the above expression to calculate the average delay per request in a cluster $k$ yields    
\begin{eqnarray}
D_k &= \frac{\rho_k}{\lambda_k} + \frac{\frac{\lambda_k,_{lc}}{\mu_{lc}^2} +  \frac{\lambda_k,_{rc}}{\mu_{rc}^2}  +  \frac{\lambda_k,_{bh}}{\mu_{bh}^2}}{1 - \rho_k}
         \label{T eqn}
 \end{eqnarray}
Based on the analysis of the delay in a single cluster, we derive the network weighted average delay per request as
\begin{eqnarray}
D =  \frac{1}{\lambda} \sum_{k=1}^{K} \lambda_k D_k,
\label{delay equation}
 \end{eqnarray}
where $\lambda =  \sum_{i=1}^{K}\lambda_i$ denotes the overall user request arrival rate in the cell. We observe from (\ref{T eqn}) that the cluster per request delay $D_k$, and correspondingly the network average delay $D$, depend on the arrival rates of the three transmission modes, which are in turn functions of the content caching scheme. Because of the limited caching capacity on mobile devices, we would like to optimize the cache placement in each cluster to minimize the network weighted average delay per request. The delay optimization problem is then formulated  as
\begin{align}
\label{optimize eqn}		
&\underset{c_k,_f}{\text{minimize}} \quad D \\
\label{optimize eqn1}
&\textrm{subject to}\quad  \sum_{f=1}^{m} c_k,_f \leq N,\\
&c_k,_f \in \{ 0, 1\},
\label{optimize eqn2}
\end{align}		
where (\ref{optimize eqn1}) and (\ref{optimize eqn2}) are the constraints that the maximum cache size is $N$ files per cluster, and the file is either cached entirely or is not cached, i.e., no partial caching is allowed. The objective function in (\ref{optimize eqn}) is not a convex function of the cache placement elements $c_{k,f} \in \{ 0, 1\}$. Moreover, this equation can be reduced to a well- known $0-1$ knapsack problem which is already proven to be NP-hard in \cite{NP-hard}. 
\begin{remark}[$N\geq m$]
\textcolor{black}{{\rm In this case, the caching problem is trivial, i.e., there are no caching constraints. For any cluster $k$, $c_k,_f=1 \quad \forall f\in \mathcal{F}$ and $\sum_{f=1}^{m} c_k,_f=m$. The optimal solution is obtained when all the files are cached in each cluster.
All the requests are served internally from the local cluster via D2D communication.} } 
\end{remark}

In the next section, we analyze the network average delay under several caching policies. We further reformulate the optimization problem in (\ref{optimize eqn}) as a well-known structure that has a locally optimal solution within a factor $(1 - e^{-1})$ of the global optimum.

\section{Proposed Caching Schemes} 	
\label{Caching Schemes}
 \subsection{Caching Popular Files (CPF)}
 
 \textcolor{black}{In each cluster, the most popular files for the users in the cluster are cached without repetition. Since popular files are different among clusters (but overlapped), applying CPF might end up replicating the same file in many clusters \cite{caire2015femtocaching}.
 We assume that the request arrival rate $\lambda_k$ is equal for all clusters.}

\subsubsection{Arrival Rate for D2D Communication}
The arrival rate of the D2D communication mode is given by
 \begin{equation}
\lambda_k,_{lc} = \lambda_k \sum_{f=\frac{k-1}{k}m_0 + 1}^{\frac{k-1}{k}m_0 + N}  P_k,_f, 	
\label{lamda1}
\end{equation}
where $\sum_{f=\frac{k-1}{k}m_0 + 1}^{\frac{k-1}{k}m_0 + N}  P_k,_f$ is the probability that the requested file is cached in the local cluster $k$, and $ f=\frac{k-1}{k}m_0 + 1$ is the most popular file index for cluster $k$. As an example, for the first cluster, $\lambda_1,_{lc} = \lambda_1 \sum_{f=1}^{N}  P_1,_f $. 

\subsubsection{Arrival Rate for Inter-cluster Communication}
\label{NaNb}	
The arrival rate of the inter-cluster communication mode is given by
  \begin{equation}
  \lambda_k,_{rc} = \lambda_k \sum_{j \in \mathcal{K}\setminus \{k\}}\sum_{f=c}^{\frac{j-1}{j}m_0 + N}  P_k,_f,    
  \label{lamda2}
\end{equation}
where $c $ is defined as max$\big(\frac{j-2}{j-1}m_0 + N + 1, \frac{j-1}{j}m_0 + 1\big)$. To explain, the inner summation $\sum_{f=c}^{\frac{j-1}{j}m_0 + N}  P_k,_f$ represents the probability that the requested file $f$ is cached in a remote cluster $j\neq k$, where the cached files in the 
$j-$th cluster are indexed from $f=\frac{j-1}{j}m_0 + 1$ to $f=\frac{j-1}{j}m_0 + N$. $c$ is defined such that a cached file in the remote clusters is counted only once when calculating $ \lambda_k,_{rc}$. The outer summation is the sum over all clusters except the local cluster $k$.


To compute the service rate of the remote cluster mode, $\mu_{rc}$, we first need to obtain the number of cooperating clusters $N_a$ since they share the cellular rate. As introduced in Section \ref{prob}, $N_a$ is a random variable representing the number of clusters served by the cellular communication whose mean is given by	
\begin{equation}
 \overline{N_a} =K \frac{\lambda_k,_{rc}}{ \lambda_k} =  K\sum_{j \in \mathcal{K}\setminus \{k\}}\sum_{f=c}^{\frac{j-1}{j}m_0 + N}  P_k,_f
 \end{equation}

 \subsubsection{Arrival Rate for Backhaul Communication}
The arrival rate of the backhaul communication mode is now calculated as
  \begin{align}
  \label{lamda3}
  \lambda_k,_{bh}& = \lambda_k\big(1- (\lambda_k,_{lc} + \lambda_k,_{rc})\big)\nonumber \\
  &= \lambda_k\Big(1- \big( \sum_{f=\frac{k-1}{k}m_0 + 1}^{\frac{k-1}{k}m_0 + N}  P_k,_f  +  \sum_{j \in \mathcal{K}\setminus \{k\}}\sum_{f=c}^{\frac{j-1}{j}m_0 + N}  P_k,_f \big)\Big)
\end{align}
$N_b$ is then obtained to calculate the backhaul service rate $\mu_{bh}$. As alluded to in the definition of $N_a$, $N_b$ is a random variable representing the number of clusters served via the backhaul link whose mean is given by
\begin{align}
   \overline{N_b} &= K \frac{\lambda_k,_{bh}}{ \lambda_k} = K\Big(1- \big( \sum_{f=\frac{k-1}{k}m_0 + 1}^{\frac{k-1}{k}m_0 + N}  P_k,_f  + \nonumber \\ 
   &\sum_{j \in \mathcal{K}\setminus \{k\}}\sum_{f=c}^{\frac{j-1}{j}m_0 + N}  P_k,_f \big)\Big)  
\end{align}
Obviously, we have $ \lambda_k = \lambda_k,_{lc}+\lambda_k,_{rc}+\lambda_k,_{bh}$. From (\ref{lamda1}), (\ref{lamda2}), and (\ref{lamda3}), the network average delay can be calculated directly from (\ref{delay equation}).  
 \textcolor{black}{The CPF scheme is computationally straightforward if the most popular content is known. 
Additionally, the CPF scheme is easy to implement in an independent manner since it is executed in a per cluster level regardless of the caching status of other clusters, which is different from the greedy algorithm proposed in the next subsection.}
However, it achieves high performance only if the popularity exponent $\beta$ is large enough, i.e., when the content popularity distribution is skewed, since a small portion of content is highly demanded which can be cached entirely in each cluster.

\subsection{Greedy Caching Algorithm (GCA)}	
 \label{mat}
In this subsection, we introduce a computationally efficient caching algorithm. We prove that the minimization problem in (\ref{optimize eqn}) can be reformulated as a \textit{minimization of a supermodular function} subject to \textit{uniform partition matroid constraints}. This structure has a greedy solution which has been proven to be locally optimal within a factor $(1 - e^{-1})$ of the optimum \cite{multi-cell,solnmono2,kowloon}.

We start with the definition of supermodular and matroid functions, then we introduce and prove some relevant lemmas.
\subsubsection{Supermodular Functions}
Let $\mathcal{S}$ be a finite ground set. The power set of the set $\mathcal{S}$ is the set of all subsets of $\mathcal{S}$, including the empty set and $\mathcal{S}$ itself. A set function $g$, defined on the powerset of $\mathcal{S}$ as $g$: $2^\mathcal{S}$$\to \mathbb{R}$, is supermodular if for any $A \subseteq B \subseteq \mathcal{S}$ and $x \in \mathcal{S}\setminus B$ we have \cite{solnmono2}
\begin{equation}
g(A\cup \{x\}) - g(A) \leq g(B\cup \{x\}) - g(B)
\end{equation}
To illustrate, let $g_A(x) = g(A\cup x) - g(A)$ denote the marginal value of an element $x \in \mathcal{S}$ with respect to a subset $A \subseteq \mathcal{S}$. Then, $\mathcal{S}$ is supermodular if for all $A \subseteq B \subseteq \mathcal{S}$ and for all $x \in \mathcal{S}\setminus B$, we have $g_A(x) \leq g_B(x)$, i.e., the marginal value of the included set is lower than the marginal value of the including set \cite{ solnmono2}.

 \subsubsection {Matroid Functions}
  Matroids are combinatorial structures that generalize the concept of linear independence in matrices \cite{solnmono2}. A matroid $\mathcal{M}$ is defined on a finite ground set $\mathcal{S}$ and a collection of subsets of $\mathcal{S}$ said to be independent. The family of these independent sets is denoted by $\mathcal{I}$ or $\mathcal{I}(\mathcal{M})$. It is common to refer to a matroid $\mathcal{M}$ by listing its ground set and its family of independent sets, i.e., $\mathcal{M} = (\mathcal{S}, \mathcal{I})$. For $\mathcal{M}$ to be a matroid, $\mathcal{I}$ must satisfy these three conditions:
\begin {itemize}
\item $\mathcal{I}$ is a nonempty set.
\item $\mathcal{I}$ is downward closed; i.e., if $B \in \mathcal{I}$ and $A \subseteq B$, then $A \in \mathcal{I}$.
\item If $A$ and  $B$ are two independent sets of  $\mathcal{I}$ and $B$ has more elements than $A$, then $\exists{e} \in B\setminus A$ such that $A \cup\{e\}\in \mathcal{I}$.

\end{itemize}
One special case is a partition matroid in which the ground set $\mathcal{S}$ is partitioned into disjoint sets $\{S_1, S_2,\dots , S_l\}$, where
 \begin{equation}
 \label{mat defn eqn}
\mathcal{I} = \{A\subseteq \mathcal{S}: |A\cap S_i|\leq k_i  \textrm{ for  all}\ i = 1, 2,\dots, l\},
\end{equation}
for some given integers $k_1, k_2, \dots, k_l$. One special case of the partition matroid is the uniform partition matroid in which $k_1 = k_2 = \dots = k_l$.

\begin{lemma}
The constraints in (\ref{optimize eqn1}) and (\ref{optimize eqn2}) can be rewritten as a uniform partition matroid on a ground set that characterizes the caching elements on all clusters.
\end{lemma}
\begin{proof}
Please see Appendix A for the proof.
\end{proof}

\begin{lemma}
The objective function in equation (\ref{optimize eqn}) is a monotone non-increasing supermodular function.
\end{lemma}
\begin{proof}
Please see Appendix B for the proof.
\end{proof}
The greedy solution for this problem structure has been proven to be locally optimal within a factor $(1 - e^{-1})$ of the optimum \cite{multi-cell,solnmono2,kowloon}.
\textcolor{black}{The greedy caching algorithm for the proposed D2D caching system with inter-cluster cooperation is illustrated in Algorithm 1, where $S_k^f$ is an element denoting the placement of file $f$ into the VCC of cluster $k$.
We first define the attributes of the system in the first line of the algorithm's pseudocode. We then initialize the cache memory of all clusters to zero. We set the number of iterations to be $NK$, which means that at each iteration, we cache one file in one cluster, resulting in caching $N$ different files in $K$ clusters after $NK$ iterations. In each iteration, all combinations of caching a file $f \in \mathcal{F}$ in a cluster $k \in \mathcal{K}$ are tried, and the network service delay is calculated. A file $f^*$ is chosen to be cached in the  $k^*$-th cluster, which achieves the highest reduction in the network service delay.}

The greedy algorithm is run at the BS level, and the BS then instructs the clusters'€™ devices to cache the files according to the output of this algorithm. The deterministic caching approach (both CPF and GCA) can only be realized if the devices stay at the same locations for many hours. Otherwise, performance obtained with the deterministic caching strategy serves as a useful upper bound for more realistic schemes \cite{caire2015femtocaching}. 
\textcolor{black}{As examples of the greedy algorithm, the authors in \cite{kowloon} showed that the problem of optimal joint caching and routing can be formulated as maximization of a monotone submodular function subject to matroid constraints, and hence can be solved by the greedy algorithm. Also, the authors in \cite{mono} showed that the delay minimization problem can be formulated as a maximization of a submodular function under matroid constraints, which can be solved by the greedy algorithm.}		

\begin{algorithm}
    \SetKwInOut{Input}{Input}
    \SetKwInOut{Output}{Output}

    \Input{$K$, $m$, $N$, $\beta$, $\overline{S}$, $R_D$, $\overline{R_{WL}}$, $\overline{R_{BH}}$;}
    \textbf{Initialization}: {$C \gets (0)_{K\times F}$}\;
    \tcc{Check if all clusters (users' memories in each cluster) are fully cached.}
       \While{$\sum_{k=1}^{K} \sum_{f=1}^{m}c_{k,f} < NK$}		
      {
        ($k^*$, $f^*$) $\gets \mathrm{argmax}_{(k, f)} D(C) - D(C\cup S_k^f)$\;
        \tcc{File achieving highest marginal value is cached.}
        $c_{k^*,f^*}=1$ \;

      }
      \Output{Cache placement $C$;}
 \caption{Greedy Caching Algorithm}
\end{algorithm}

\section{Throughput Analysis}
We have analyzed the per request average delay from the network perspective under different caching schemes. In this section, we conduct the per request throughput and throughput scaling analysis. We first characterize the per request throughput from the queuing theory perspective based on the analytical results of previous sections, then study the scaling of the average sum throughput when the number of files asymptotically goes to infinity. 

\subsection{{per request Throughput Analysis}}
In this subsection, we first formulate a condition on the traffic demand for the network to be stable, then we study the throughput per request from the cluster perspective. 
As introduced in Section \ref{sec}, the content size $S_f$ is assumed to have an exponential distribution with mean $\overline{S}$ [bits]. For a cluster $k \in \mathcal{K}$ whose users' traffic is modeled as an MPSQ with three serving processors (transmission modes), the number of users' requests in the queue that matches the $j$-th transmission mode is denoted by $x_j$, where $j \in \mathbb{D} := \{M_{lc}, M_{rc}, M_{bh}\}$. Denote $\textbf{x} = (x_j)_{j \in \mathbb{D}}$ as the vector counting the numbers of users' requests in the queue for each transmission mode $j \in \mathbb{D}$.

The process $\{X(t); t \geq 0\}$ describing the number of users' requests served by the three serving processors (transmission modes) is then a continuous-time Markov process \cite{queue_32}. This process has a discrete state space $\mathbb{N^D}$ and admits the following generator \cite{queue_32}: 		
\[
\begin{cases}
               q(\textbf{x}, \textbf{x} + \epsilon_j) = \lambda_{k,j},     \quad  \quad  \quad &\textbf{x} \in \mathbb{N^D}, j \in \mathbb{D}, \\ 
              q(\textbf{x}, \textbf{x} - \epsilon_j) = \frac{R_j}{\overline{S}} \frac{x_j}{x_\mathbb{D}}, \quad  \quad  &\textbf{x} \in \mathbb{N^D}, j \in \mathbb{D}, x_j >0,
            \end{cases}
\]
where $\epsilon_j$ designates the vector of $\mathbb{N^D}$ having coordinate 1 at position j and 0 elsewhere, and $x_\mathbb{D}:= \sum_{j \in \mathbb{D}}x_j$. The first term of the above generator, $q(\textbf{x}, \textbf{x} + \epsilon_j)$, accounts for the arrival of a request that matches the $j-$th transmission mode while the second term, $q(\textbf{x}, \textbf{x} - \epsilon_j)$, accounts for serving a request by the $j-$th transmission mode.

Let $\textbf{X} = (X_{lc}, X_{rc}, X_{bh})$ be the vector counting the number of users' requests of each transmission mode at the steady state, and let $X_\mathbb{D} := \sum_{j \in \mathbb{D}}X_j$ be the total number of requests in the queue at the steady state. 
The average traffic demand $\zeta_j$ [bps] of each transmission mode $j \in \mathbb{D}$ in the $k-$th cluster is defined as \cite{heten_mpsq}  
 \begin{equation}
\zeta_j = \lambda_{k,j} \overline{S},
\end{equation}
and the total traffic demand per cluster is then given by
 \begin{equation}
\zeta = \sum_{j \in \mathbb{D}}\zeta_j
\end{equation}
We now obtain the cluster critical traffic demand, beyond which the MPSQ is no longer stable. The constraint (\ref{rho_k}) that limits the traffic intensity $\rho_k$ from the above to one can be rewritten as		
\begin{align}
   \rho_k &= \frac{\lambda_{k,lc}}{R_D/\overline{S}} + \frac{\lambda_{k,rc}}{\overline{R_{WL}}/\overline{S}} + \frac{\lambda_{k,bh}}{\overline{R_{BH}}/\overline{S}} \leq 1,\nonumber \\
  &\frac{\zeta_{lc}}{R_{D}} + \frac{\zeta_{rc}}{\overline{R_{WL}}} + \frac{\zeta_{bh}}{\overline{R_{BH}}} \leq 1,
  \label{critical}
\end{align}
by multiplying both sides by $\zeta$ and rearranging the terms, we get
\begin{align}
   &\zeta \leq \frac{\zeta}{(R_D^{-1}\zeta_{lc} + \overline{R_{WL}}^{-1}\zeta_{rc} + \overline{R_{BH}}^{-1}\zeta_{bh})}, \nonumber \\
   &\zeta \leq \zeta_c, 
  \label{critical_zeta}
\end{align}
where $\zeta_c$ [bps] is the critical traffic demand per cluster, beyond which the MPSQ loses its stability.
\begin{lemma}
The steady state distribution of the total number of users' requests in the MPSQ modeling the users' traffic follows a geometric distribution with parameter $p =1 - \zeta/\zeta_c$.
\end{lemma}
\begin{proof}
This result can be deduced from \cite{queue_32} and the references therein, and the proof is omitted in this paper to avoid repetition.
\end{proof}
As a direct result from Lemma 3, the mean number of total users' requests in the MPSQ at the steady state is given by
\begin{equation}
\overline{N_{q}} = E[X_\mathbb{D}] = \frac{p}{1-p} = \frac{\zeta}{\zeta_c - \zeta}
\end{equation}

At the steady state, the queue throughput is equal to the traffic demand $\zeta$. Hence, the average throughput per request is defined as the ratio of the given queue throughput and the average number of users' requests, i.e.,
\begin{equation}
\overline{r} = \frac{\zeta}{E[X_\mathbb{D}]} = \zeta_c - \zeta
\end{equation}
\subsection{Throughput Scaling Analysis}
\label{Asymptotic}
We conduct the scaling analysis of the average sum throughput when the number of files grows asymptotically to infinity, i.e., $m \to \infty$. 
We first define the outage probability for our proposed D2D cooperative caching system and then compare it with a clustered D2D caching system without inter-cluster cooperation \cite{thropt_outage}. The obtained formula of the outage probability is further approximated and then exploited in the throughput scaling analysis.

In the following, we shall implicitly ignore the non-integer effects when they are irrelevant for the scaling laws. For example, recalling that the network has node density $n$ and it is divided into $K$ clusters, the number of users per cluster after integer rounding is denoted as $y$. Next, we conduct the analysis for the CPF scheme. Since the backhaul rate is considered much smaller than the rate of cellular and D2D communications, we assume that the throughput from the backhaul communication is negligible as compared to the cellular and D2D throughput.
 \subsubsection{Outage Probability}	
 \label{outage}
 For a reference clustered D2D caching network without inter-cluster cooperation \cite{thropt_outage}, the probability of no outage is defined as the probability that a randomly chosen user $u$ can download a requested file from nearby users in the same cluster \cite{thropt_outage}. 
 Conversely, a user $u$ is said to be in outage when its requested file is not cached within the allowed transmission range (i.e., not cached in a neighbor user in the same cluster). 
 In our cooperative clustered model, a user $u$ is said to be in outage when the requested file is neither stored in the local cluster nor any of the remote clusters. We denote this outage probability as $p_o$, which also represents the percentage of users who are in outage in relation to the total number of users;  the probability of no outage is then denoted as $1 - p_o$. 
 
As stated before, the number of users per cluster, denoted as $y$, equals $(n/K)$. 
In addition, the probability of no outage, $1 - p_o$, can be calculated by determining the probability that a randomly chosen user $u$ in cluster $k$ is served via the local cluster or the remote cluster modes.
The probability of no outage is therefore expressed as the sum of two terms, the first term is corresponding to the probability of serving requests from the local cluster, and the second term is the probability of being served from a remote cluster.
From (\ref{lamda1}) and (\ref{lamda2}), and under the assumption of the CPF scheme, the probability of no outage is given by
\begin{equation}
\label{lamda44}
1 - p_0 = \sum_{f=\frac{k-1}{k}m_0 + 1}^{\frac{k-1}{k}m_0 + My}  P_k,_f + \sum_{j \in \mathcal{K}\setminus \{k\}}\sum_{f=c}^{\frac{j-1}{j}m_0 + My}  P_k,_f,    
\end{equation}
where $M$ is the maximum user cache size in files (our default is $M = 1$), and $c$ is defined in (\ref{lamda2}).
Substituting $P_{k,f}$ from (\ref{popularity eqn}), we obtain the result
\begin{equation}
\label{lamda4}
1 - p_0 = \frac{ \sum_{f=\frac{k-1}{k}m_0 + 1}^{\frac{k-1}{k}m_0 + My}f^{-\beta}} { \sum_{i=1}^{m}i^{-\beta} } + \sum_{j \in \mathcal{K}\setminus \{k\}}\frac{ \sum_{f=c}^{\frac{j-1}{j}m_0 + My}f^{-\beta}  }{  \sum_{i=1}^{m} i^{-\beta} } 
\end{equation}
Due to the symmetry between clusters in terms of the cache content, cluster cache size, and the probability of being served from a remote cluster, we continue with the assumption that the user $u$ is being served from the first cluster (i.e., $k=1$) and the remote clusters (the potential cooperating clusters) are from $k=2$ to $k=K = n/y$. 
\begin{align}
\label{lamda5}
1 - p_0 &= \frac{ \sum_{f=1}^{My}f^{-\beta}} { \sum_{i=1}^{m}i^{-\beta} } + \sum_{j=2}^{\frac{n}{y}}\frac{ \sum_{f=c}^{\frac{j-1}{j}m_0 + My  }f^{-\beta}  }{  \sum_{i=1}^{m} i^{-\beta} }\nonumber \\
&= \frac{ \sum_{f=1}^{My}f^{-\beta}} { \sum_{i=1}^{m}i^{-\beta} } + \frac{1}{ \sum_{i=1}^{m} i^{-\beta}}\sum_{j=2}^{\frac{n}{y}}\sum_{f=c}^{\frac{j-1}{j}m_0 + My}f^{-\beta}
\end{align}

We now aim at deriving an approximated version of (\ref{lamda5}) by replacing the summations with approximated integrals from \cite{zeta}, and then the obtained result is used later in the throughput scaling analysis. We have two approximations from \cite{zeta},
\begin{equation}
\label{approx1}
\sum_{i=1}^{q}i^{-\alpha} \approx \int_{i}^{q+1}x^{-\alpha}dx = \frac{(q+1)^{1-\alpha} - 1}{1 - \alpha},
\end{equation}
and 
\begin{align}
\label{approx2}
\sum_{i=w+1}^{q-1}i^{-\alpha} &\approx \int_{w}^{q}x^{-\alpha}dx - \frac{w^{\alpha} + q^{\alpha}}{2}, \nonumber \\
&= \frac{q^{1 - \alpha} - w^{1 - \alpha}}{1 - \alpha} - \frac{w^{\alpha} + q^{\alpha}}{2}
\end{align}
The above approximations are quite tight for small values of the popularity exponent, e.g., when $\beta < 1$. Substituting (\ref{approx1}) and (\ref{approx2}) into (\ref{lamda5}) yields
\begin{align}
&1 - p_0 \approx \frac{\frac{1}{1 - \beta}(My+1)^{1-\beta} - \frac{1}{1 - \beta}} { \frac{1}{1 - \beta}(m+1)^{1-\beta} - \frac{1}{1 - \beta} } + \nonumber \\
&\frac{1}{\frac{1}{1 - \beta}(m+1)^{1-\beta} - \frac{1}{1 - \beta} }\sum_{j=2}^{\frac{n}{y}}\Big(
\frac{  \big( \frac{j-1}{j}m_0 + My + 1\big)^{1-\beta} - \big(c'\big)^{1- \beta}   }{1 - \beta} \nonumber \\
&- \frac{  \big(c'\big)^{\beta} + \big( \frac{j-1}{j}m_0 + My + 1\big)^{\beta}  }{2}\Big),\nonumber \\
& =  \overline{ p}_{0,nc} + \overline{p}_{0,wc}
\label{approx3}
\end{align}
where $c' =$max$(\frac{j-i}{j}m_0, \frac{j-2}{j-1}m_0 + My)$, and $ \overline{ p}_{0,nc}$, $\overline{p}_{0,wc}$ represent respectively the probability of no outage for a non-cooperative system and the improvement (increase) in the probability of no outage due to the inter-cluster cooperation. In Fig.~\ref{outage_fig}, we plot the outage probability of our proposed system with inter-cluster cooperation compared to a reference system without inter-cluster cooperation. We note that as the number of users per cluster increases, the outage probability correspondingly decreases. That is attributed to the fact that the probability of obtaining the requested files from the local cluster increases with the number of users per cluster.			
\begin{figure}
\begin{center}
\includegraphics[width=3.0in]{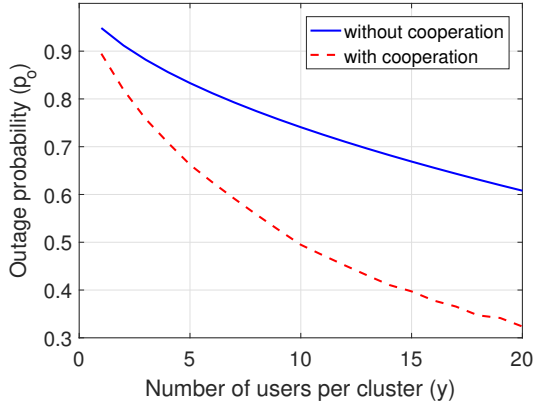}
\caption { Outage probability of a D2D clustered caching system with cooperation compared to a reference system without cooperation \cite{thropt_outage} ($m = 108, n = 120, M = 1, m_0 = 60, \beta=0.5$).}
\label{outage_fig}
\end{center}
\end{figure}

\subsubsection{Throughput Scaling Analysis}
We now express the network average sum throughput, denoted as $T_{sum}^{avg}$ (bps), as a function of the system parameters, namely, number of users, library size, and popularity exponent. Based on the assumed interference model, only one D2D link can be active at any time in each cluster. Whenever there is an active D2D link within a
cluster, we say the cluster is good.\footnote{In this article, and different from \cite{thropt_outage}, we neglect the inter-cluster interference. We assume that any cluster can be active whenever there is a scheduled D2D link, regardless of the activity of all other clusters. This assumption makes the calculated throughput an upper bound for the actual throughput.}
We also assume that a D2D link is scheduled in any cluster whenever the opportunity arises, i.e., if the user to be served in a cluster requests a file not cached locally, the request is then served via the appropriate transmission mode (remote cluster or backhaul modes), meanwhile, another D2D link is scheduled inside the cluster. In addition to the D2D throughput, there is the cellular throughput from the remote cluster mode. So the instantaneous throughput $T_{sum}$ (bps) can be written as			
\begin{align}
T_{sum} &= \textrm{D2D throughput} + \textrm{Cellular throughput},\nonumber 
\end{align}
and the average sum throughput is obtained from 
\begin{align}
T_{sum}^{avg} &= R_{D2D} E[L] + R_{WL} P_{rc},
\end{align}
where $L$ is the number of active D2D links, $E[L]$ is the expected number of active D2D links, which is approximately the expected number of good clusters, and $P_{rc}$ is the probability of occurrence of cooperation between clusters. For notational simplicity, we henceforth substitute $R_{D2D}$ by $C$ (bps) and $R_{WL}$ by $k_1C$ (bps), where $k_1 < 1$. 
\begin{align}
\label{small}
T_{sum}^{avg} &= C \big( E[L] + k_1 P_{rc}\big),\nonumber \\
& \leq C \big( E[L] + k_1\big),
\end{align}
where the above inequality holds because $P_{rc} $ is a  probability and cannot be greater than one. In particular, $P_{rc}$ is tight with its upper bound for a large number of clusters and relatively uniform popularity distribution (i.e., not skewed). 
In the sequel, we calculate the expected number of good clusters $E[L]$.

Up to now, the cell is divided into $K = (n/y)$ virtual clusters, each of them with $y$ uniformly distributed users. As mentioned before, a cluster is good if at least one user requests a file that can be served from the locally cached content via D2D communication. Conversely, a cluster is not good if all $y$ users in the same cluster cannot serve their requests from the locally cached content, which occurs with probability $p_{0,nc}^{y}$ \cite{Scaling_Behavior}, where $p_{0,nc} = 1 - \overline{ p}_{0,nc}$  
is the probability that a randomly chosen user $u$ in any cluster can not obtain a requested file from nearby users in the same cluster. The probability of having a good cluster is then $1 - p_{0,nc}^{y}$. Therefore, we have the following
\begin{align}
\label{approx33}
E[L] = \frac{n}{y}(1 - p_{0,nc}^{y})
\end{align}
Substituting $p_{0,nc}$ from (\ref{approx3}), and (\ref{approx33}) into (\ref{small}) yields
\begin{align} 
\label{small1}
T_{sum}^{avg} &\leq C \Big(\frac{n}{y}(1 - p_{0,nc}^{y}) + k_1\Big), \nonumber \\
&= C \frac{n}{y} \Big(1 - (1 - \frac{\frac{1}{1 - \beta}(My+1)^{1-\beta} - \frac{1}{1 - \beta}} { \frac{1}{1 - \beta}(m+1)^{1-\beta} - \frac{1}{1 - \beta} })^{y}\Big) + k_1C
\end{align}
Similar to \cite{thropt_outage} and \cite{Scaling_Behavior}, we define the quantity 
\begin{align} 
\label{gamma}
\gamma = \frac{1 - \beta}{2 - \beta},
\end{align}
where $\gamma$ changes from $0$ to $\frac{1}{2}$ when $\beta$ changes from $1$ to $0$. These ranges of $\beta$ and $\gamma$ are interesting for the scaling analysis since they are reasonable in practice \cite{thropt_outage}. 
In the following, we conduct the scaling analysis for the regime when $y$ changes sublinearly with $m$ \cite{thropt_outage}, i.e, $y =  \rho m^{\gamma}$ for some constant $\rho$, and $\gamma \leq \frac{1}{2}$. 
We analyze the scaling of the upper bound for $T_{sum}^{avg}$ when $m$ asymptotically grows to infinity. 
Substituting $y = \rho m^{\gamma}$ into (\ref{small1}) yields
\begin{align} 
\label{scaling}
T_{sum}^{avg} &\leq C \frac{n}{\rho m^{\gamma}} \Big(1 - \big(1 - \frac {\frac{1}{1 - \beta}(My+1)^{1-\beta} - \frac{1}{1 - \beta} }  { \frac{1}{1 - \beta}(m+1)^{1-\beta} - \frac{1}{1 - \beta}  }\big)^{\rho m^{\gamma}}\Big) \nonumber \\
 &+ k_1C, \nonumber \\
&= C \frac{n}{\rho m^{\gamma}} \Big(1 - \big(1 - M^{1-\beta} \rho^{1-\beta} F^{(1-\beta)(\gamma-1)} \big)^{\rho m^{\gamma}}\Big) + \nonumber \\
&+ k_1C, \nonumber \\
&\overset{(a)}{=} C \frac{n}{\rho m^{\gamma}} \Big(1 - \big(1 - M^{1-\beta} \rho^{1-\beta} m^{-\gamma} \big)^{\rho m^{\gamma}}\Big) + k_1C, 
\end{align}
where (a) follows by using $(1-\beta )(\gamma -1) = -\gamma$, then we have
\begin{align} 
&T_{sum}^{avg} \leq  k_1C + C \frac{n}{\rho m^{\gamma}} \Big(1 - \nonumber \\
&\big(\big(1 - \rho^{1-\beta} M^{1-\beta} m^{-\gamma} \big)^{\rho^{-(1 - \beta)}M^{-(1 - \beta)}m^{\gamma}}\big)^{\rho^{2 - \beta} M^{-(1 - \beta)}}\Big), \nonumber \\
&\overset{(b)}{=} k_1C + \frac{C}{\rho } \Big(1 - \big(e^{-1} \big)^{\rho^{2 - \beta} M^{-(1 - \beta)}}       \Big)      \frac{n}{m^{\gamma} }, 
\end{align}
where (b) follows from $\lim_{x\to\infty} (1 - x^{-1})^{x} = e^{-1}$, then we have
\begin{align}
T_{sum}^{avg} &\leq \frac{C}{\rho } \Big(1 - e^{-\rho^{2 - \beta} M^{-(1 - \beta)}}       \Big)      \frac{n}{m^{\gamma} }  + k_1C, \nonumber \\
&= \Theta\Big(\frac{n}{m^{\gamma} }\Big) + O(1)
\end{align}
This result shows that: 
\begin{itemize}
	\item As the library size $m$ increases, the upper bound for $T_{sum}^{avg}$ decreases, since the probability of having active D2D links (good clusters) decreases. 
	\item As $\gamma$ increases, corresponding to the decrease of the popularity exponent $\beta$, the upper bound for $T_{sum}^{avg}$ vanishes more rapidly with the library size $m$.
	\item The upper bound for $T_{sum}^{avg}$ scales linearly with the number of users $n$.\footnote{We use the standard Landau notation: $g(n) = O(g(n))$ denotes $g(n) \leq c_1g(n)$ and $g(n) = \Theta(g(n))$ denotes $k_1g(n) \leq g(n) \leq k_2g(n)$, where $c_1$, $k_1$, and $k_2$ are real constants $>0$.}
\end{itemize}
 

The average sum throughput is plotted against the number of users per cluster $y$ in Fig.~\ref{thropt_vs_G}, for different values of $\beta$. We observe that there is an optimal value of $y$ at which the throughput is maximized. 
First, the throughput increases with the cluster size $y$. Then, as the cluster size increases, the outage probability decreases owing to the higher cache size per cluster. However, for larger cluster size, the throughput starts to decrease owing to the decrease in the number of clusters associated with the larger cluster size. 

\begin{figure}
\begin{center}
\includegraphics[width=3.0in]{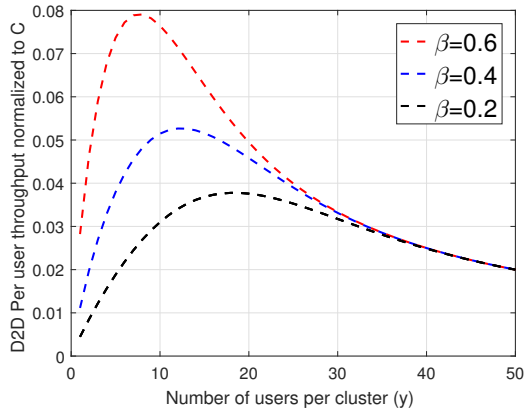}		
\caption {D2D per-user  throughput of the cooperative system is plotted against the number of users per cluster $y$ at different values of the popularity exponent $\beta$ (parameters as in \cite{thropt_outage}, $n = 10,000$ users$, m = 1000$ files $, m_0=200$ files).}
 \label{thropt_vs_G}
\end{center}
\end{figure}

\section{Numerical Results}
\label{sim}
In this section, we evaluate the performance of our proposed inter-cluster cooperative architecture using simulation and analytical results. Results are obtained with the following parameters: $\lambda_k = 0.5$ requests/sec, $m_0=60$ files, $m = 108$ files, $\overline{S} = 4$ Mbits, $K = 5$ clusters, $n = 25$ users, $M=4$ files, and $N= 20$ files. $R_{WL} =50$ Mbps and $R_{BH} = 5$ Mbps  as in \cite{multi-cell}. For a typical D2D communication system with transmission power of 20 dBm, transmission range of 10 m, and free space path loss model as in \cite{basic_principle}, we have $R_{D} = 120$ Mbps.

\begin{figure}
	\begin{center}
		\includegraphics[width=3.0in]{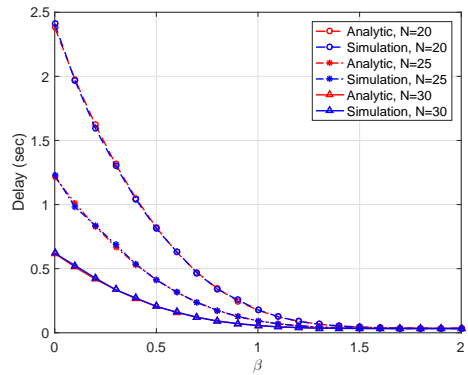}		
		\caption {Network average delay versus popularity exponent $\beta$ under the CPF scheme.}
		\label{delay}
	\end{center}
\end{figure}

\begin{figure}
	\begin{center}
		\includegraphics[width=3.0in]{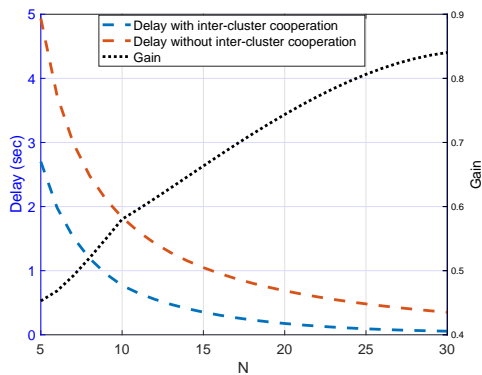}		
		\caption {Network average delay (left hand side y-axis) and gain (right hand side y-axis) vs cluster cache size $N$.}
		\label{cache_size}
	\end{center}
\end{figure}

In Fig.~\ref{delay}, we verify the accuracy of the analytical results of the network average delay under the CPF with inter-cluster cooperation. The theoretical and simulated results for the network average delay under the CPF scheme are plotted together, and they are consistent. We see that the network average delay is significantly improved by increasing the cluster cache size $N$. Moreover, as $\beta$ increases, the average delay decreases. This is attributed to the fact that a small portion of content forms most of the requests that can be cached locally in each cluster and delivered via high data rate D2D communication.

In the following, we evaluate and compare the performance of various caching schemes.  In Fig.~\ref{cache_size}, our proposed inter-cluster cooperative caching system is compared with a D2D caching system without cooperation under the CPF scheme. For a D2D caching system without cooperation, requests for files that are not cached in the local cluster are downloaded directly from the core network. For the sake of concise comparison, we define the delay reduction gain as
\begin{equation}
\textrm{Gain} = 1 - \frac{\textrm{Delay with inter-cluster cooperation}}{\textrm{Delay without inter-cluster cooperation}}	
\end{equation} 
Fig.~\ref{cache_size} shows that, for a small cluster cache size, the delay reduction (gain) of our proposed inter-cluster cooperative caching is higher than 45\% with respect to a D2D caching system without inter-cluster cooperation and greater than 80\% if the cluster cache size is large. 

\begin{figure}
	\begin{center}
		\includegraphics[width=3.0in]{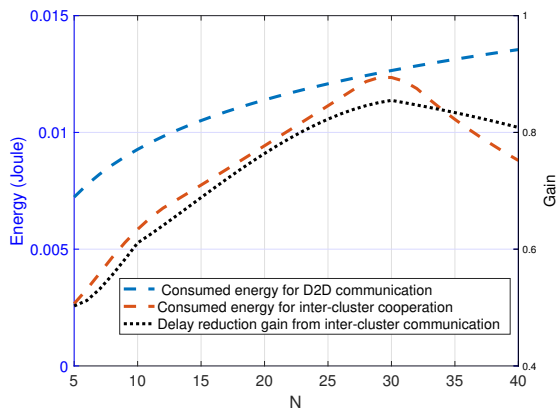}
		\caption {Energy consumption per cluster during the local and remote cluster transmissions (left hand side y-axis) and the gain attained from inter-cluster cooperation (right hand side y-axis) vs cluster cache size $N$.}
		\label{cache_size11}
	\end{center}
\end{figure}

To show the energy-delay reduction gain tradeoff among the devices, in Fig.~\ref{cache_size11}, we plot the per-cluster energy consumption during the local and remote cluster modes and the gain attained from inter-cluster cooperation against the cluster cache size $N$. $P_{lc} = 20$ dBm, and $P_{rc} = 23$ dBm denote respectively the transmission power in the local cluster and remote cluster modes. In each transmission mode, the energy per request is the transmission power times the transmission duration. The transmission duration is given by the ratio of file size over the transmission rate. 
We see that the consumed energy during the local cluster transmission, i.e., D2D communication, monotonically increases with the cluster cache size $N$. 
With the increasing of $N$, more requests are served via the local cluster mode $M_{lc}$. 
For the consumed energy during the remote cluster transmission, we see that it initially increases with $N$, then it decreases, and the same behavior is observed for the delay-reduction gain. This can be interpreted as follows. When $N$ increases, the number of requests served from the remote clusters increases since the remote clusters' VCCs increase. When $N$ becomes much larger, the local cluster cache becomes sufficiently large to serve most of the requests, as opposed to being served by the remote cluster mode.

\begin{figure*}
\centering
  \subfigure[Network average delay vs request arrival rate for three caching schemes, CPF, GCA, and RC. ]{\includegraphics[width=3.0in]{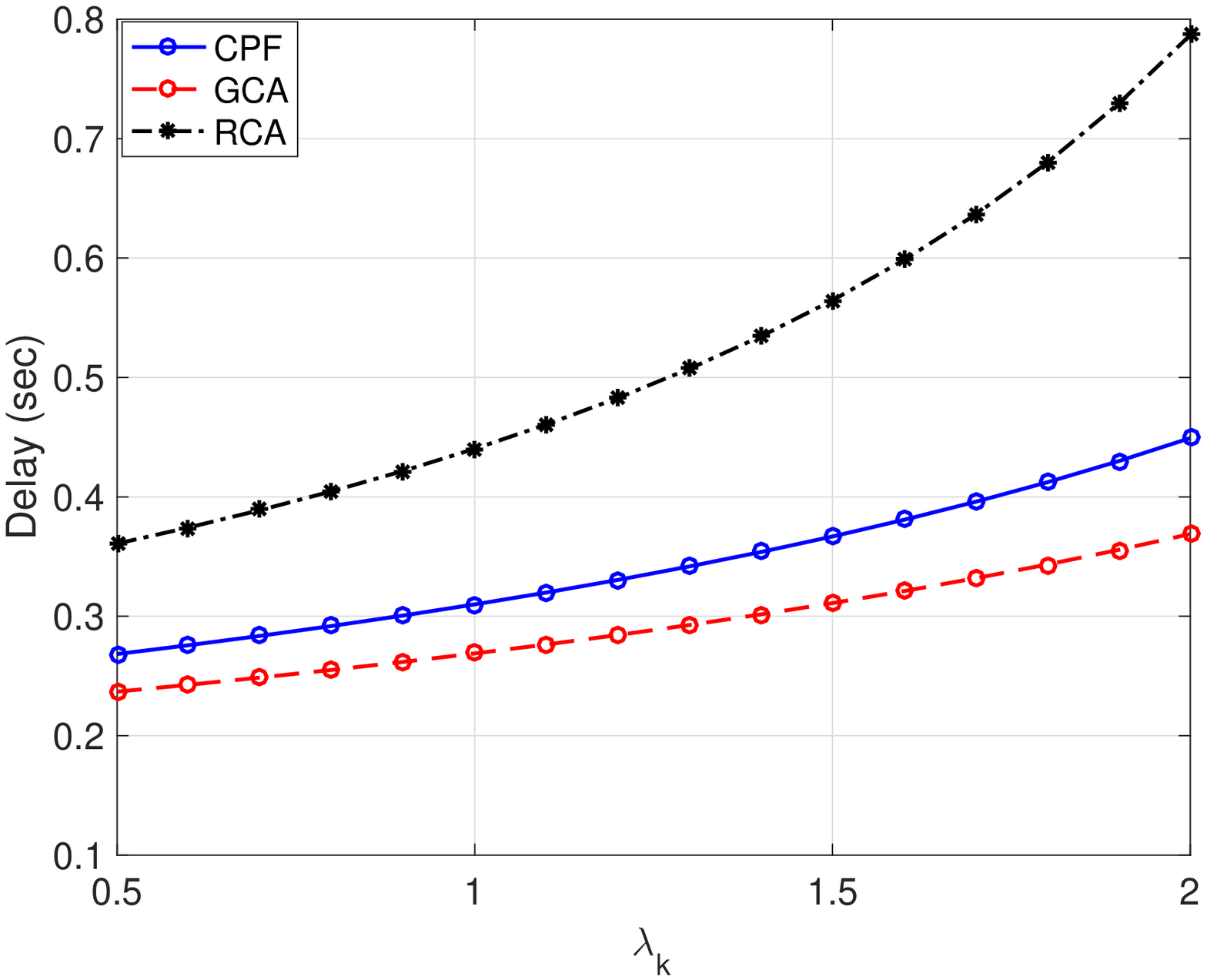}
\label{locally optimal delay vs lambda}}
\subfigure[Network average delay vs popularity exponent for three caching schemes, CPF, GCA, and RC.]{\includegraphics[width=3.0in]{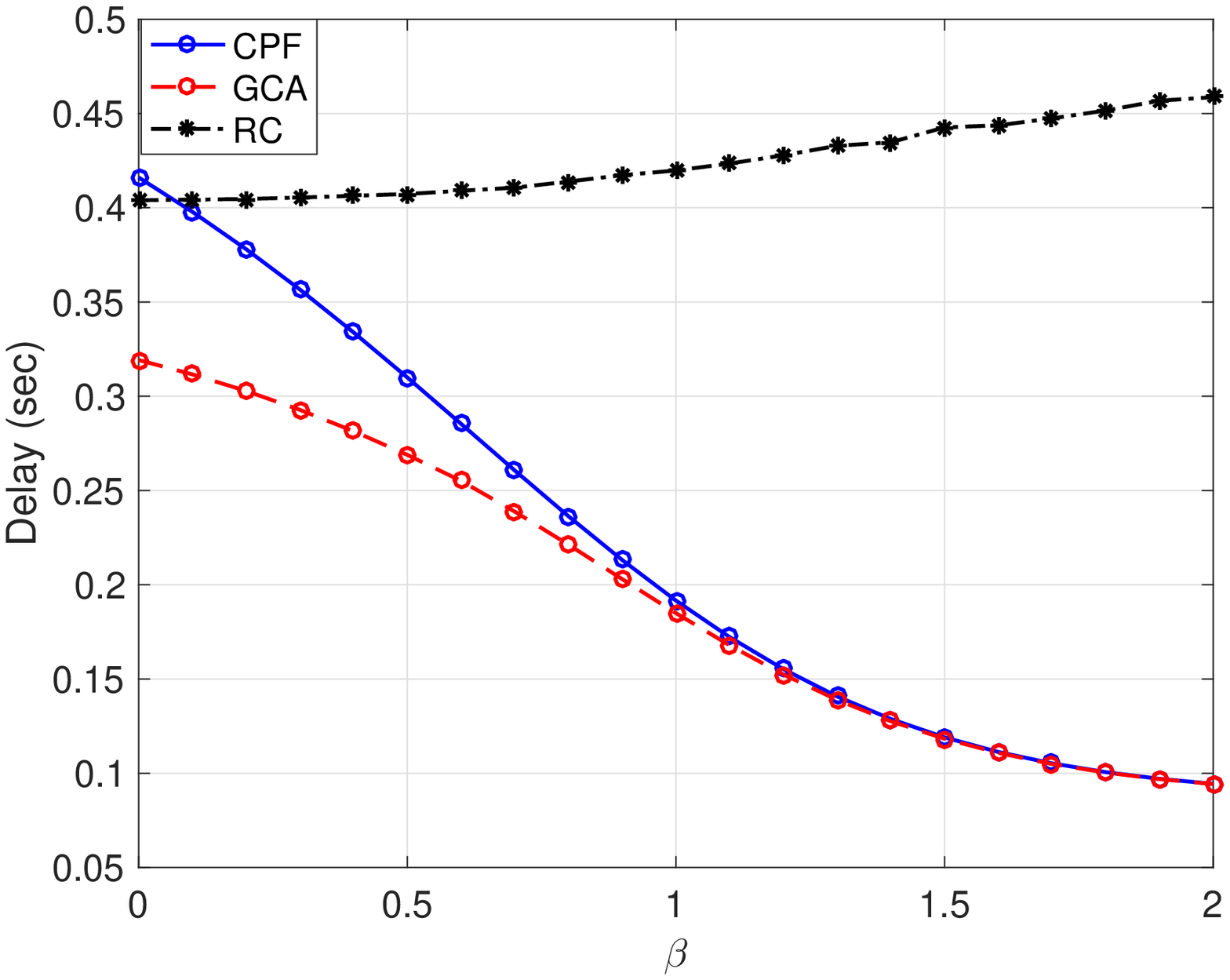}		
  \label{CPF and GCA}}
\caption{Evaluation and comparison of the average delay for the proposed caching schemes and random caching for various system parameters ($R_{D}= 50$ Mbps, $\overline{R_{WL}}= 15$ Mbps, $\overline{R_{BH}}= 10$ Mbps, $N=20$, $\beta=0.5$ for (a) and $\lambda_k = 0.5$ requests/sec for (b)).}	
\label{all}
\end{figure*}
For comparison purposes, Fig.~\ref{all} shows the average delay for the proposed caching schemes and random caching against various system parameters. Fig.~\ref{locally optimal delay vs lambda} shows the network average delay plotted against the request arrival rate $\lambda_k$ for three content placement techniques, namely, GCA, CPF, and random caching (RC).\footnote{Here, we adopt different transmission rates from \cite{basic_principle} and \cite{multi-cell} to provide insights on the difference between the caching schemes, otherwise, the GCA is far superior to the other schemes, with negligible delay.} In RC, content stored in clusters are randomly chosen from the file library. The most popular files are cached in the CPF scheme, and the GCA works as illustrated in Algorithm 1.
We see that the average delay for all content caching strategies increases with $\lambda_k$ since a larger request rate increases the probability of a longer waiting time for each request. It is also observed that the GCA, which is locally optimal, achieves significant performance gains over the CPF and RC solutions for the above setup. 			
Fig.~\ref{CPF and GCA} shows that the GCA is superior to the CPF only for small values of the popularity exponent $\beta$. If the popularity exponent $\beta$ is high enough, CPF and GCA will achieve the same performance. 
When  $\beta$ increases, the CDF of the Zipf's distribution becomes more skewed. This implies that only a smaller portion of the files is highly demanded by the devices. The lower the number of files requested by the devices, the higher the probability of having such files cached in the clusters' VCCs. If all these files are cached locally in each cluster, the global minimum solution for the delay minimization problem is attained. This interpretation explains why when $\beta$ increases, the CPF and GCA solutions converge to the global optimal solution. 
We also note that the CPF and RC schemes roughly achieve the same delay when $\beta = 0$. This stems from the fact that with $\beta = 0$, all files have equal popularity, and correspondingly, CPF is equivalent to RC.
Moreover, RC fails to reduce the delay as $\beta$ increases, since caching files at random results in a low probability of serving the requested files from local clusters. 

Next, we turn our attention to the throughput results in Fig.~\ref{all_1}.
Fig.~\ref{throp_vs_beta} plots the throughput per request as a function of the popularity exponent $\beta$ for the three caching schemes. It is shown that the per request throughput monotonically increases  with $\beta$ for the CPF and GCA schemes, and shows a slight decrease for the RC scheme. When $\beta$ increases for the GCA and CPF, the locally stored files form most of the users' requests that can be delivered via high rate D2D communication. Conversely, for the RC scheme, which caches the files uniformly at random, the probability of having the requested files cached in the local clusters slightly decreases when the popularity of files becomes skewed (higher $\beta$). Due to the resulting lower probability of serving the requests from the local clusters, the throughput per request, in turn, slightly decreases owing to the lower probability of activating D2D links. 
In Fig.~\ref{throp_vs_N}, the throughput per request is plotted against the cluster cache size $N$ for the three caching schemes. It is noticed that for all the caching schemes, the per request throughput is improved with the cluster cache size, and the GCA achieves the highest throughput. This can be explained by the fact that, with large cluster cache size, there is a high opportunity of exchanging cached content via the local cluster mode that exploits the high rate of the D2D communication.   

\begin{figure*}
\centering
  \subfigure[The per request throughput vs popularity exponent for three caching schemes, CPF, GCA, and RC.]{\includegraphics[width=3.0in]{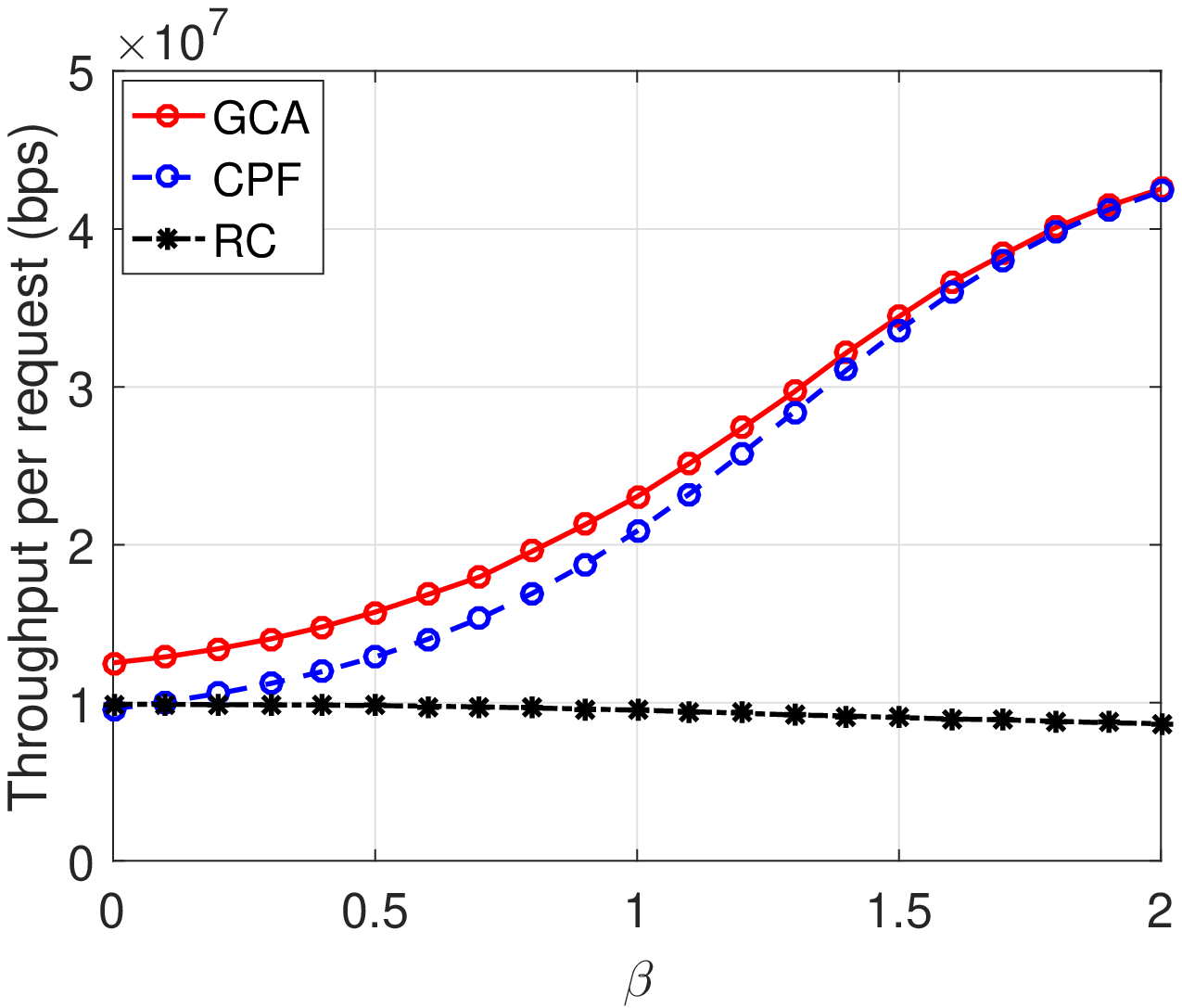}
\label{throp_vs_beta}}
\subfigure[The per request throughput vs cluster cache size for three caching schemes, CPF, GCA, and RC.]{\includegraphics[width=3.0in]{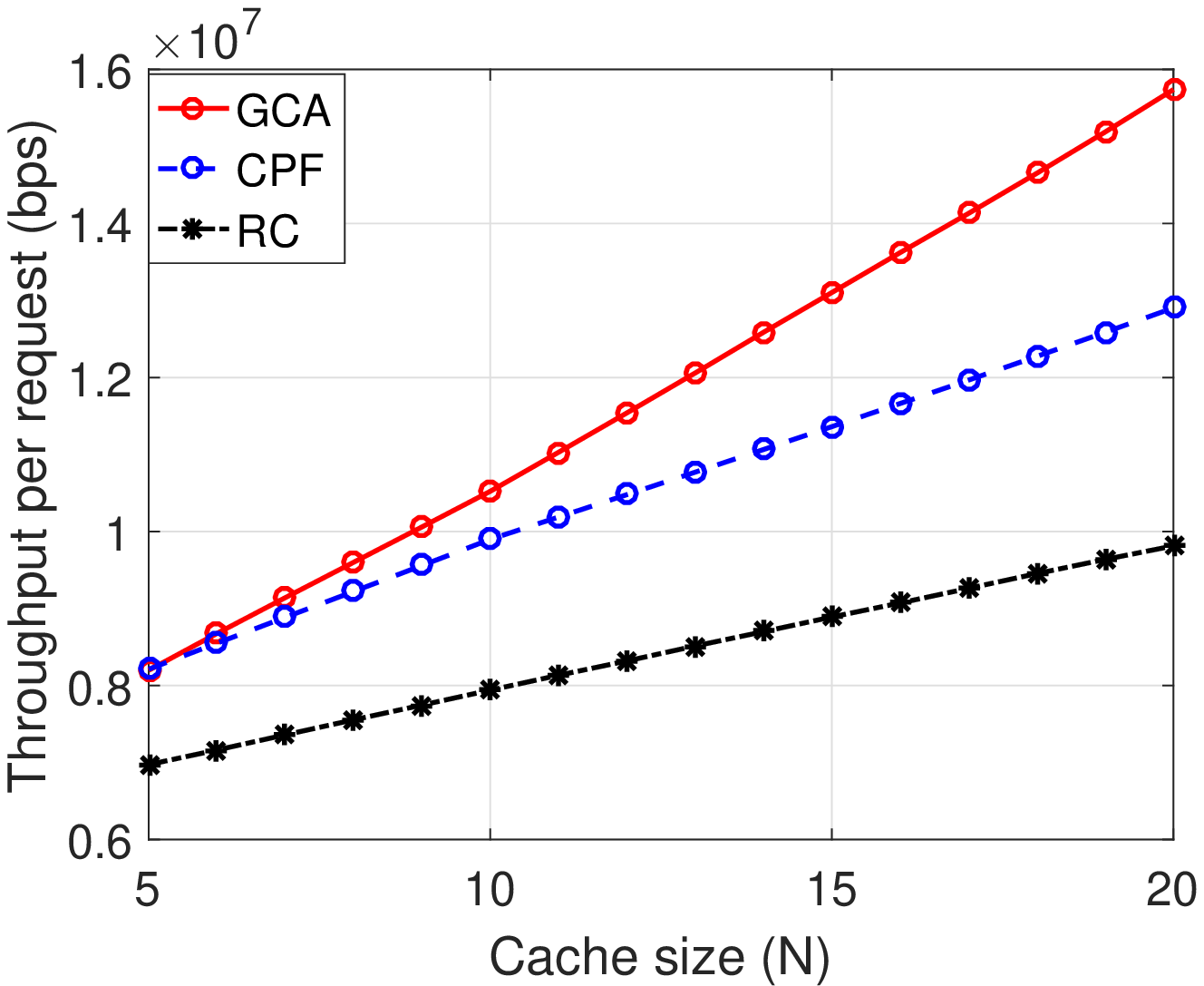}		
  \label{throp_vs_N}}
\caption{Evaluation and comparison of the per request throughput for the proposed caching schemes and random caching for various system parameters ($R_{D}= 50$ Mbps, $\overline{R_{WL}}= 15$ Mbps, $\overline{R_{BH}}= 10$ Mbps, $\lambda_k = 0.5$ requests/sec).}	
\label{all_1}
\end{figure*}

\section{Conclusion}
In this work, we propose a novel D2D caching architecture to reduce the network average delay. We study a cellular network consisting of one SBS and a set of users. The cell is divided into a set of equally-sized virtual clusters, where the users in the same cluster exchange cache content via D2D communication, while the users in different clusters cooperate by exchanging their cache content via cellular transmission. We formulate the delay minimization problem in terms of the content cache placement. However, the problem is NP-hard and obtaining the optimal solution is computationally hard. We then propose two content caching policies, namely, caching popular files and greedy caching. By reformulating the delay minimization problem as a minimization of a non-increasing supermodular function subject to uniform partition matroid constraints, we show that it could be solved using the proposed GCA scheme within a factor $(1 - e^{-1})$ of the optimum. Moreover, we conduct the throughput analysis to investigate the behavior of the average throughput per request under different caching schemes.
We study the scaling behavior of the average sum throughput when the library size asymptotically grows to infinity and show that  the network average sum throughput decreases with the library size increase, and the rate of this decrease is controlled by the popularity exponent. We verify our analytical results by means of extensive simulations and the results show that the network average delay could be reduced by 45\%-80\% by allowing inter-cluster cooperation.
\begin{appendices}
\section{Proof of lemma 1}
We define the ground set that describes the cache placement elements in all clusters as
\begin{equation}
\mathcal{S} = \{s_1^1, ...,s_k^f, ..., s_k^m, ..., s_K^1, ...,  s_K^m\}
\label{set eqn}
\end{equation}
where $s_k^f$ is an element denoting the placement of file $f$ into the VCC of cluster $k$. This ground set can be partitioned into $K$ disjoint subsets $\{S_1, S_2, ..., S_K\}$, where $S_k =  \{s_k^1, s_k^2, ..., s_k^m\}$ is the set of all files that might be placed in the VCC of cluster $k$.

Let us express the cache placement by the adjacency matrix $\textbf{X} = [x_k,_f ]_{K\times m} \in \{0, 1\}_{K\times m} $. Moreover, we define the corresponding cache placement set $A \subseteq \mathcal{S}$ such that $s_k^f \in A$ if and only if $x_{k,f}=1$. Hence, the constraints on the cache capacity of the VCC of cluster $k \in {\mathcal{K}}$ can be expressed as $A \subseteq \mathcal{S}$, where
\begin{equation}
\mathcal{H} = \{A\subseteq \mathcal{S}: |A\cap S_k| \leq N  \textrm{ for  all}\ k = 1, \dots, K\}
\label{matroid eqn}
\end{equation}
The above expression is derived directly from the constraint that the maximum cache size per cluster is $N$ files, i.e., $\sum_{f=1}^{m}x_k,_f \leq N$. Comparing $\mathcal{H}$ in (\ref{matroid eqn}) with the definition of partition matroid in (\ref{mat defn eqn}), it is clear that our constraints form a partition matroid with $l = K$ and $k_i = N$. Additionally, since $k_i = N$ for all $i = \{1, 2,\dots, K\}$, it is easy to see that our constraints also form a uniform partition matroid. 
This proves Lemma 1.

\section{Proof of lemma 2}
We consider two cache placement sets $A$ and  $A'$, where $A \subset A'$. For a certain cluster $k \in \mathcal{K}$, we consider adding the caching element $s_k^f$ $\in \mathcal{S}\setminus A'$ to both placement sets. This means that a file $f$ is added to cluster $k$, where the corresponding cache placement element has not been placed in either $A$ or $A'$. The marginal value of adding an element $s_k^f$ to a set is defined as the change in the file download time after adding this element to the set. The average download time for a file $f$ with mean size $\overline{S}$ is $\frac {\overline{S}}{R_{D}}$, $\frac {\overline{S}}{R_{WL}/N_a}$, or $\frac {\overline{S}}{R_{BH}/N_b}$ if the file is obtained from the local cluster, a randomly chosen remote cluster, or the backhaul, respectively.
For our work, we assume that $\frac{R_{WL}}{{N_a}} > \frac{R_{BH}}{{N_b}}$ always holds. For the sake of simplicity, we replace $\frac{R_{WL}}{{N_a}}$ and $\frac{R_{BH}}{{N_b}}$ with their averages, $\overline{R_{WL}}$ and $\overline{R_{BH}}$, respectively. Now, the aggregate transmission rate assumption is $R_D > \overline{R_{WL}} > \overline{R_{BH}}$.

For $D_k$ in (\ref{T eqn}) to be a supermodular function, the difference in the marginal values between the two sets $A$ and $A'$ must be non-positive.
For a user $u$ belonging to cluster $k$ and requesting content $f \in \mathcal{F}$, we distinguish between these different cases:

\begin{enumerate}

  \item According to placement $A'$, user $u$ obtains file $f$ from a remote cluster $j'$, i.e., $s_{j'}^f \in A'$ and $j' \neq k$. In this case, the marginal value with respect to $A'$ is
\begin{align}
G(A' \cup \{s_k^f\}) - G(A') = 0 						
\end{align}
  According to placement $A$, user $u$ obtains file $f$ from a remote cluster $j$, i.e., $s_{j}^f \in A$, again the marginal value is zero. However, if $s_{j}^f\notin A$, the marginal value is given by
 \begin{align}
G(A \cup \{s_k^f\}) - G(A) = P_{k,f} \Big( \frac {\overline{S}}{\overline{R_{WL}}} -  \frac {\overline{S}}{\overline{R_{BH}}}\Big) 
\end{align}
\item In this case, we assume that $s_i^f = s_k^m$, i.e., the requested file $f$ is cached in cluster $k$. According to placement $A'$, user $u$ obtains file $f$ from the local cluster $k$. Hence, the marginal value is given  by
 \begin{align}
G(A' \cup \{s_i^f\}) - G(A') = P_{k,f} \Big( \frac {\overline{S}}{R_{D}} -  \frac {\overline{S}}{\overline{R_{WL}}}\Big)	
\end{align}
According to placement $A$, user $u$ obtains file $f$ from a remote cluster $j$ when $s_{j}^f \in A$, again the marginal value is given  by
  \begin{align}
G(A \cup \{s_i^f\}) - G(A) = P_{k,f} \Big( \frac {\overline{S}}{{R_{D}}} -  \frac {\overline{S}}{\overline{R_{WL}}}\Big)	
\end{align}
However, if $s_{j}^f\notin A$, the marginal value is written  as
 \begin{align}
G(A \cup \{s_i^f\}) - G(A) = P_{k,f} \Big( \frac {\overline{S}}{R_{D}} -  \frac {\overline{S}}{\overline{R_{BH}}}\Big) 
\end{align}
\end {enumerate}
Accordingly, the difference in marginal values between $A$ and $A'$ in all cases is
\begin{align}
G(A \cup \{s_i^f\}) - G(A) - (G(A' \cup \{s_i^f\}) - G(A')) \leq 0				
\end{align}
It is clear that $g(A) \leq g(A')$ for $A \subseteq A'  \subseteq \mathcal{S}$, or equivalently, $g(A) - g(A') \leq 0$. From the definition of supermodularity, it is clear that the delay per request in the $k-$th cluster, $D_k$, is a supermodular set function. The weighted sum of supermodular functions is also a supermodular function \cite{solnmono2}, and so the network average delay $D$ in (\ref{optimize eqn}) is a supermodular function. For the monotone non-increasing property, it is intuitive to see that the delay will never increase by caching new files. Hence, Lemma 2 proves that problem (\ref{optimize eqn}) is a monotonically non-increasing supermodular set function minimized under uniform partition matroid constraints.
\end{appendices}

\bibliographystyle{IEEEtran}
\bibliography{bibliography}
\begin{IEEEbiography}[{\includegraphics[width=1in,height=1.25in,clip,keepaspectratio]{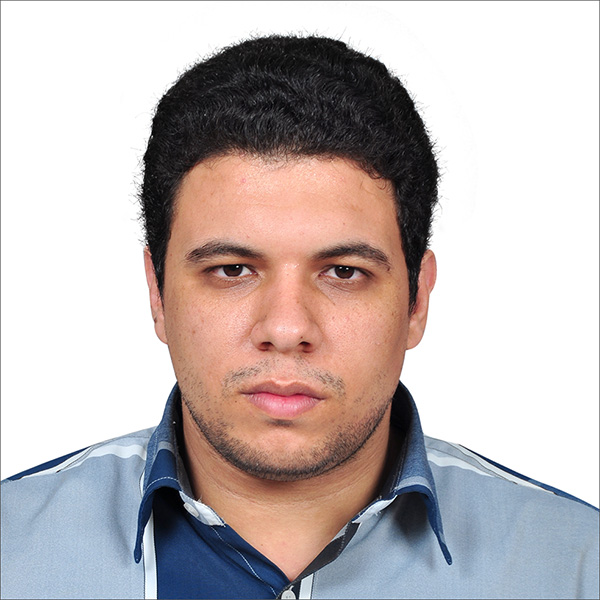}}]{Ramy Amer} received the MSc degree in Electrical Engineering from Alexandria University, Egypt, in 2016. Since 2016, he is pursuing his PhD degree in wireless communication, especially, wireless caching, under the supervision of Dr Nicola Marchetti and Dr M. Majid Butt at CONNECT centre, Trinity College Dublin, Ireland. His research interests include cross-layer design, Cognitive Radio, Energy Harvesting, Stochastic Geometry, and Reinforcement Learning. Prior to CONNECT, he was also an assistant lecturer for the switching department, National Telecommunication Institute of Egypt (NTI), Cairo, Egypt, where he conducted professional training both at the national and international levels. He is also certified as a Cisco instructor and has other Cisco data and voice certificates. He worked as a part-time instructor for many national and international training centres, e.g., New-Horizon Egypt and Fast-Lane KSA. 
\end{IEEEbiography}
\begin{IEEEbiography}[{\includegraphics[width=1in,height=1.25in,clip,keepaspectratio]{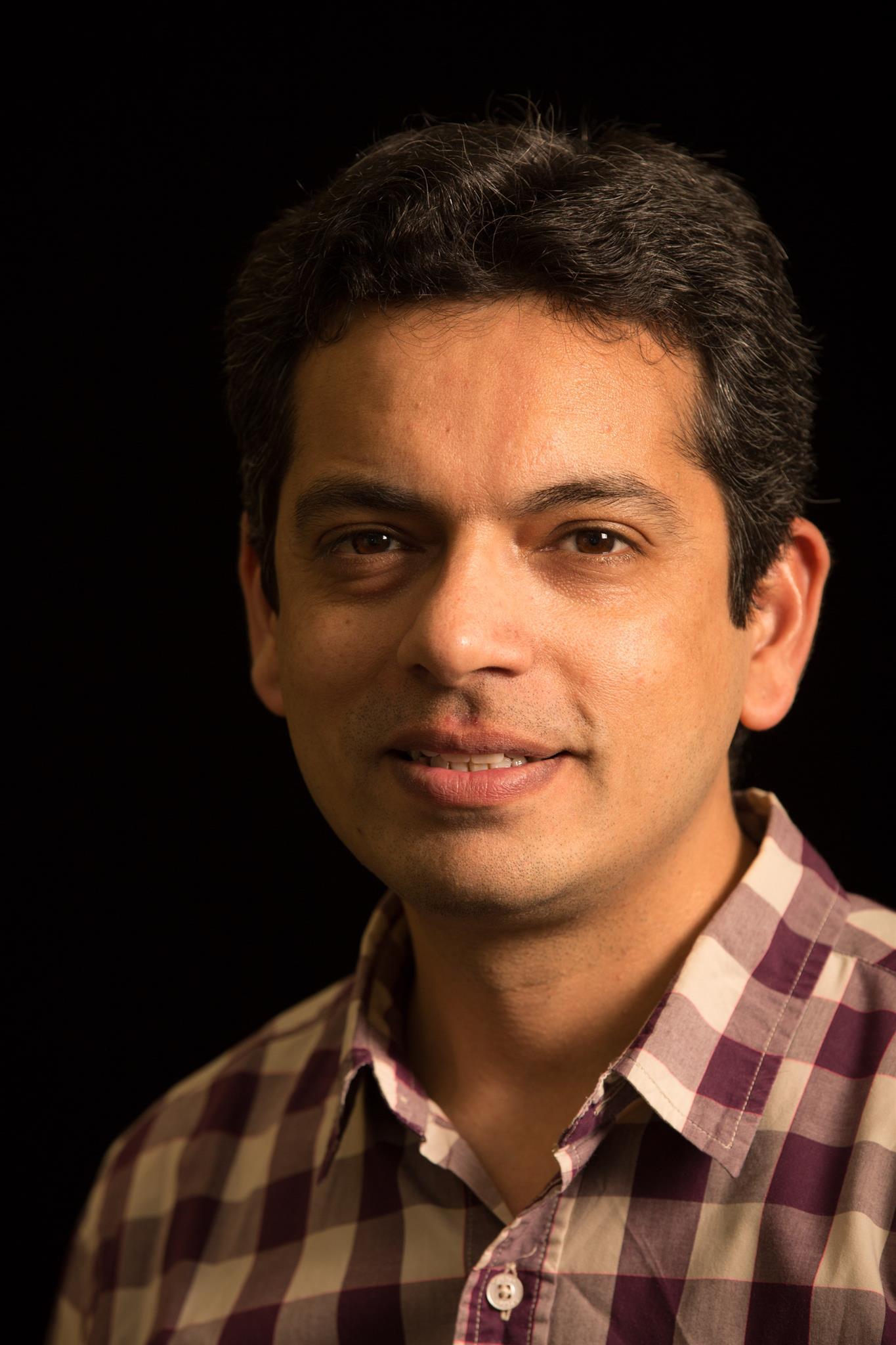}}]{M. Majid Butt} (S'07 -- M'10 -- SM'15) received the MSc degree in Digital Communications from Christian Albrechts University, Kiel, Germany, in 2005, and the PhD degree in  Telecommunications from the Norwegian University of Science and Technology, Trondheim, Norway, in 2011. He is an Assistant Professor at University of Glasgow as well as an adjunct Assistant Professor at Trinity College Dublin, Ireland. Before that, he has held senior researcher positions at Trinity College Dublin, Ireland and Qatar University. He is recipient of Marie Curie Alain Bensoussan postdoctoral fellowship from European Research Consortium for Informatics and Mathematics (ERCIM). He held ERCIM postdoc fellow positions at Fraunhofer Heinrich Hertz Institute, Germany, and University of Luxembourg. Dr. Majid's major areas of research interest include communication techniques for wireless networks with focus on radio resource allocation, scheduling algorithms, energy efficiency and cross layer design. He has authored more than 50 peer reviewed conference and journal publications in these areas. He has served as TPC chair for various communication workshops in conjunction with IEEE WCNC, ICUWB, CROWNCOM, IEEE Greencom and Globecom. He is a senior member of IEEE and serves as an associate editor for IEEE Access journal and IEEE Communication Magazine since 2016.
\end{IEEEbiography}
\begin{IEEEbiography}[{\includegraphics[width=1in,height=1.25in,clip,keepaspectratio]{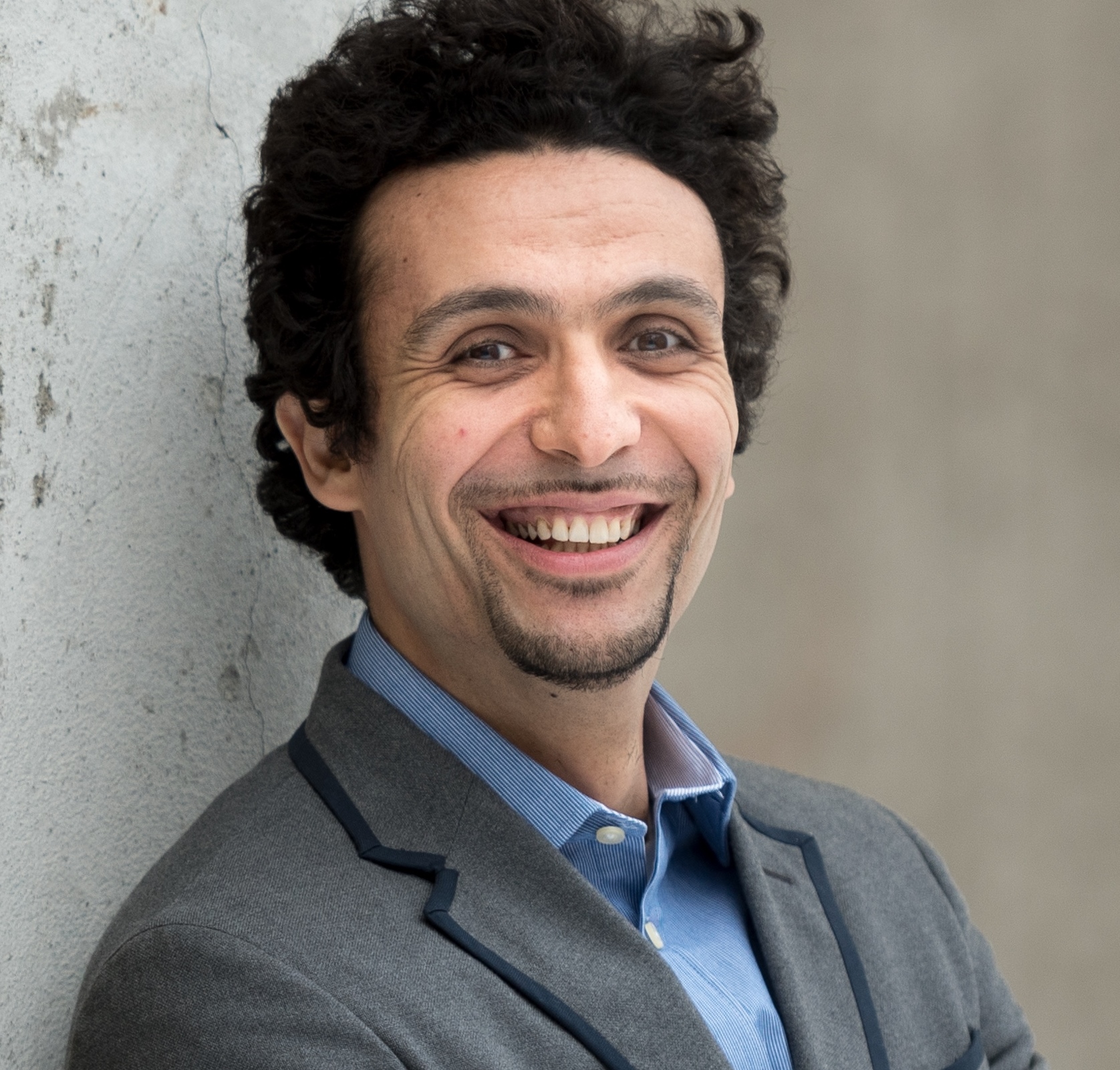}}]{Mehdi Bennis} (S'07 -- AM'08 -- SM'15) received his M.Sc. degree in electrical engineering jointly from EPFL, Switzerland, and the Eurecom Institute, France, in 2002. He obtained his Ph.D. from the University of Oulu in December 2009 on spectrum sharing for future mobile  cellular systems. Currently he is an associate professor at the  University of Oulu and an Academy of Finland research fellow.  His main research interests are in radio resource management, heterogeneous networks, game theory, and machine learning  in 5G networks and beyond. He has co-authored one book  and published more than 200 research papers in international  conferences, journals, and book chapters. He was the recipient  of the prestigious 2015 Fred W. Ellersick Prize from the IEEE  Communications Society, the 2016 Best Tutorial Prize from  the IEEE Communications Society, the 2017 EURASIP Best  Paper Award for the Journal of Wireless Communications and  Networks, and the all-University of Oulu research award.
\end{IEEEbiography}
\begin{IEEEbiography}[{\includegraphics[width=1in,height=1.25in,clip,keepaspectratio]{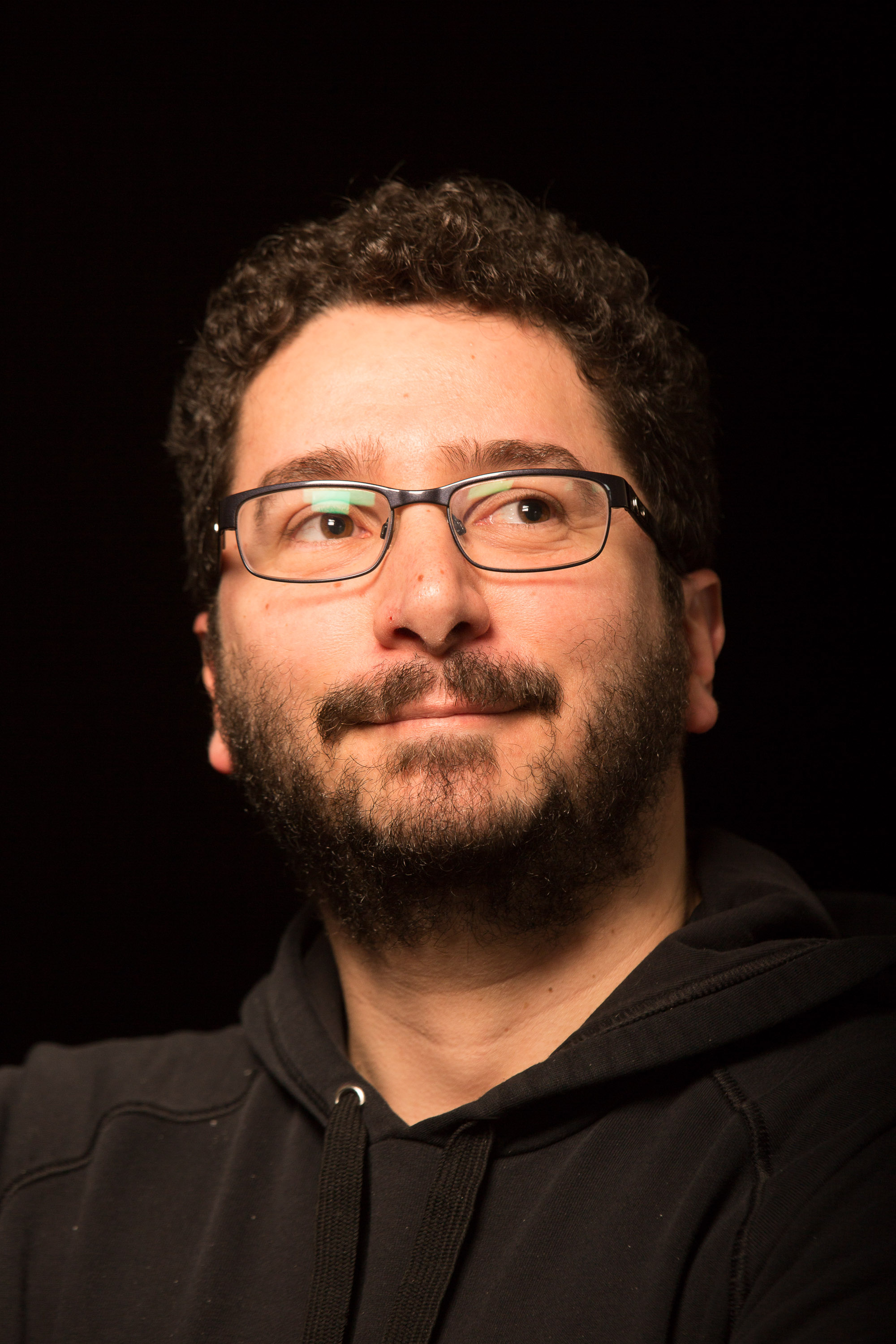}}]{Nicola Marchetti} Dr. Nicola Marchetti is currently Assistant Professor in Wireless Communications at Trinity College Dublin, Ireland. He performs his research under the Trinity Information and Complexity Labs (TRICKLE) and the Irish Research Centre for Future Networks and Communications (CONNECT). He received the PhD in Wireless Communications from Aalborg University, Denmark in 2007, and the M.Sc. in Electronic Engineering from University of Ferrara, Italy in 2003. He also holds an M.Sc. in Mathematics which he received from Aalborg University in 2010. His collaborations include research projects in cooperation with Nokia Bell Labs and US Air Force Office of Scientific Research, among others. His research interests include Adaptive and Self-Organizing Networks, Complex Systems Science for Communication Networks, PHY Layer, Radio Resource Management. He has authored 110 journals and conference papers, 2 books and 7 book chapters, holds 2 patents, and received 4 best paper awards.
\end{IEEEbiography}

\end{document}